\documentclass[11pt]{article}
\usepackage{amsmath,amsthm,amsfonts,amssymb}
\usepackage{cancel}
\usepackage{fullpage}
\usepackage{liyang}
\usepackage{framed}
\usepackage{verbatim}
\usepackage{enumitem}
\usepackage{array}
\usepackage{multirow}
\usepackage{afterpage}

\usepackage[normalem]{ulem}

\usepackage{caption} 

\usepackage[usenames,dvipsnames]{xcolor}

\usepackage{todonotes}

\makeatletter
\newtheorem*{rep@theorem}{\rep@title}
\newcommand{\newreptheorem}[2]{
\newenvironment{rep#1}[1]{
 \def\rep@title{#2 \ref{##1}}
 \begin{rep@theorem}\itshape}
 {\end{rep@theorem}}}
\makeatother
\theoremstyle{plain}



\def\verbose{0}

\def\colorful{0}

\ifnum\colorful=1
\newcommand{\violet}[1]{{\color{violet}{#1}}}
\newcommand{\orange}[1]{{\color{orange}{#1}}}
\newcommand{\blue}[1]{{{\color{blue}#1}}}
\newcommand{\red}[1]{{\color{red} {#1}}}

\newcommand{\gray}[1]{{\color{gray}{#1}}}

\fi
\ifnum\colorful=0
\newcommand{\violet}[1]{{{#1}}}
\newcommand{\orange}[1]{{{#1}}}
\newcommand{\blue}[1]{{{#1}}}
\newcommand{\red}[1]{{{#1}}}

\newcommand{\gray}[1]{{{#1}}}

\fi

\usepackage{boxedminipage}

\newreptheorem{theorem}{Theorem}
\newtheorem*{theorem*}{Theorem}
\newreptheorem{lemma}{Lemma}
\newreptheorem{proposition}{Proposition}
\newtheorem*{noclaim*}{Claim}

\newcommand{\uhr}{\upharpoonright}

\newcommand{\encode}{\mathrm{encode}}
\newcommand{\weight}{\mathrm{weight}}

\newcommand{\proj}{\mathrm{proj}}

\newcommand{\Sipser}{\mathsf{Sipser}}

\newcommand{\acz}{\mathsf{AC^0}}

\newcommand{\STCONN}{\mathit{STCONN}}

\newcommand{\SkewedSipser}{\mathsf{SkewedSipser}}

\newcommand{\ourd}{d} % This is "our d", the # of levels of AND_2's in our modified Sipser functions

\newcommand{\restletter}{\kappa}  % New greek letter used for restrictions in the proof of main theorem

\begin{document}

%\title{$\tilde{\text{O}}$ptimal small-depth lower bounds for small distance connectivity}

\title{\LARGE Near-optimal small-depth lower bounds\\ for small distance connectivity}

\author{
Xi Chen\thanks{xichen@cs.columbia.edu. Supported in part by
  NSF grants CCF-1149257 and CCF-1423100.}\\ Columbia University
\and
Igor C. Oliveira\thanks{oliveira@cs.columbia.edu.}\\ Columbia University
\and
Rocco A. Servedio\thanks{rocco@cs.columbia.edu.  Supported in part by NSF grants CCF-1319788 and CCF-1420349.}\\ Columbia University 
\and Li-Yang Tan\thanks{liyang@cs.columbia.edu.  Part of this research was done while visiting Columbia University.} \\ Toyota Technological Institute
}

\maketitle

\begin{abstract}
We show that any depth-$d$ circuit for determining whether an $n$-node graph has an $s$-to-$t$ path of length at most $k$ must have size $n^{\Omega(k^{1/d}/d)}$.  The previous best circuit size lower bounds for this problem were $n^{k^{\exp(-O(d))}}$ (due to Beame, Impagliazzo, and Pitassi~\cite{BIP:98}) and $n^{\Omega((\log k)/d)}$ (following from a recent formula size lower bound of Rossman~\cite{Rossman:13}). Our lower bound is quite close to optimal, since a simple construction gives depth-$d$ circuits of size $n^{O(k^{2/d})}$ for this problem (and strengthening our bound even to $n^{k^{\Omega(1/d)}}$ would require proving that undirected connectivity is not in $\mathsf{NC^1}.$)%\rnote{Igor mentioned we might want to quantify ``for all $k$ and $d$'' here in the abstract.  Since here in the abstract we're not quantifying either the $k$'s for which our lower bound result holds, nor the $k$'s for which BIP/Rossman hold, I think it's OK for us to be a bit vague here and not do any quantification in the abstract (of course we'll be completely precise in the intro).}

Our proof is by reduction to a new lower bound on the size of small-depth circuits computing a skewed variant of the ``Sipser functions'' that have played an important role in classical circuit lower bounds~\cite{sipser1983,Yao:85,Hastad:86}.  A key ingredient in our proof of the required lower bound for these
Sipser-like functions is the use of \emph{random projections}, an extension of random restrictions which were recently employed in \cite{RST:15}.  Random projections allow us to obtain sharper quantitative bounds while employing simpler arguments, both conceptually and technically, than in the previous works \cite{Ajtai:89,BPU:92,BIP:98,Rossman:13}.
\end{abstract}

 \thispagestyle{empty}
\newpage

% !TEX root =  main.tex

\setcounter{page}{1}
\section{Introduction}\label{sec:intro}

Graph connectivity problems are of great interest in theoretical computer science, both from an {algorithmic} and a computational complexity perspective. The ``$st$-connectivity,'' or $\STCONN$, problem --- given an $n$-node graph $G$ with two distinguished vertices $s$ and $t$, is there a path of edges from $s$ to $t$? --- plays a particularly central role.  One  longstanding question is whether any improvement is possible on Savitch's $O((\log n)^2)$-space algorithm \cite{Savitch:70}, based on ``repeated squaring,'' for the \emph{directed} $\STCONN$ problem; since this problem is complete for $\mathsf{NL}$, any such improvement would show that $\mathsf{NL} \subseteq \mathsf{SPACE}(o(\log^2 n))$, and hence would have a profound impact on our understanding of non-deterministic space complexity.  
\ifnum\verbose=1
For \emph{undirected} graphs, the landmark result of Reingold \cite{Reingold:08} showed that there is a deterministic log-space algorithm for $st$-connectivity.   
\fi
Wigderson's survey \cite{Wigderson:92connectivity} provides a now somewhat old, but still very useful, overview of early results on connectivity problems.

In this paper we consider the ``small distance connectivity'' problem $\STCONN(k(n))$ which~is defined as follows.  The input is the adjacency matrix of an undirected $n$-vertex graph $G$
 which has two distinguished vertices $s$ and $t$, and the 
problem is to determine whether $G$ contains a path of length at most $k(n)$ from $s$ to $t$.  We study this problem from the perspective of small-depth circuit complexity; for a given depth 
$d$ (which may depend on $k$), we are interested in the size of unbounded fan-in depth-$d$ 
circuits of $\AND$, $\OR$, and $\NOT$ gates that compute $\STCONN(k(n))$.  (As several authors \cite{BIP:98, Rossman:13} have observed, the directed and undirected versions of the $\STCONN(k(n))$ problems are essentially equivalent via a simple reduction that converts a directed graph into a layered undirected graph; for simplicity we focus on the undirected problem in this paper.)

An impetus for this study comes from the above-mentioned question about Savitch's algorithm.  As noted by Wigderson~\cite{Wigderson:92connectivity}, a simple reduction shows that if Savitch's algorithm is optimal, then for all $k$ polynomial-size unbounded fan-in circuits for $\STCONN(k(n))$ %which are of polynomial size 
must have depth $\Omega(\log k)$.  By giving lower bounds on the size of small-depth circuits for $\STCONN(k(n))$, Beame, Impagliazzo, and Pitassi~\cite{BIP:98} have shown that depth $\Omega(\log \log k)$ is required for $k(n) \leq \log n$, and more recently Rossman~\cite{Rossman:13} has shown that depth $\Omega(\log k)$ is required for $k(n) \leq \log \log n$.  These bounds for restricted ranges of $k$ motivate further study of the circuit complexity of small-depth circuits for $\STCONN(k(n)).$  Below we give a more thorough discussion of both upper and lower bounds for this problem, before presenting our new results.

%\rnote{\gray{One option would be to add a paragraph here saying ``Motivation for this study comes from the question of whether Savitch is optimal.  As BIP pointed out, any improvement on repeated squaring for any distance k would result in an improved algorithm for the general directed st-connectivity problem...''  We could talk about how if Savitch is optimal then depth $\Omega(\log k)$ is really required for all $k$.  This is known due to Rossman for very small $k$, but for larger $k$ much less is known.  (Might be a bit of a drawback to doing it this way because we don't actually prove
%$\Omega(\log k)$ for any $k$.)
%
%Another option would be not to have such a paragraph but just tack a sentence on to the previous paragraph saying something like ``This question has been studied by several researchers [Ajtai, BIP, Rossman]; we move to a description of their results below.''
%
%Probably better to go the first route, though it will take a little care, and also have a last sentence of the paragraph segueing into the prior results that follow.}}

\ignore{
%We consider the problem $\STCONN(k(n))$ which is defined as follows.  The input is an undirected graph $G$
%on $n$ vertices (given as a Boolean adjacency matrix) which has two distinguished vertices $s$ and $t$.  The 
%problem is to determine whether $G$ contains a path of length at most $k(n)$ from $s$ to $t$.  For a given depth 
%$d$ (which may depend on $k$), we are interested in the size of unbounded fan-in depth-$d$ 
%circuits of $\AND$, $\OR$ and $\NOT$ gates that compute $\STCONN(k(n))$. \rnote{Elaborate on how this problem is fundamental, etc.  Maybe/probably mention,
%like Ben does,  that the directed and undirected versions of this problem are essentially equivalent, unlike general connectivity.  In general amplify the background section,
%perhaps plundering \cite{BIP:98} and \cite{Rossman:13}'s introductions.}
}

\subsection{Prior results}\label{sec:subsection_prior_results} 
\noindent {\bf Upper bounds (folklore).} A natural approach to obtain efficient circuits for $\STCONN(k(n))$ is by
repeated squaring of the input adjacency matrix.  If $x_{i,j}$ is the input variable that takes value 1 if 
edge $\{i,j\}$ is present in the input graph, then the graph contains a path of length at most 2 from $i$ to $j$ if and only if
the depth-2 circuit 
$
\bigvee_{k=1}^n (x_{i,k} \wedge x_{k,j})
$
is satisfied (assuming that $x_{i,i}=1$ for every $i$).  Iterating this construction yields a circuit of size $\poly(n)$ and depth $2 \log k$ that computes
$\STCONN(k(n))$, whenever $k$ is a power of two.
\if\verbose=1
(For arbitrary values of $k$, one can easily adapt this construction by viewing it as a sequence of recursive calls of depth $2 \lceil \log k \rceil$ that checks for paths of length $k$ between nodes $u$ and $v$ via potential paths of length $\lceil k/2 \rceil$ and $\lfloor k/2 \rfloor$ passing through intermediate nodes.)
\fi

For smaller depths, a natural extension of this approach leads to the following construction. Let $G_0$ be the input graph. For every pair of nodes $u,v$ in $G_0$, check by exhaustive search for~paths of length at most $t = k^{1/d}$ connecting these nodes. (We assume that $k^{1/d}$ is an integer in order to avoid unnecessary technical details.)
 Note that this can be done simultaneously for every pair of nodes by a (multi-output) depth-2 $\OR$-of-$\AND$s circuit of size $n^{t + O(1)}$. Let $G_1$ be a new graph that has an edge between $u$ and $v$ if and only if a path of length at most $t$ connects these nodes.
\ifnum\verbose=1
 Intuitively, large distances in $G_0$ are shorter by a factor of roughly $t$ in $G_1$. 
 \fi
 In general, if we start with $G_0$ and repeat this procedure $d$ times, we obtain a sequence of graphs $G_0, G_1, \ldots, G_d$ for which the following holds: $G_i$  has an edge between nodes $u$ and $v$ if and only if they are connected by a path of length at most $t^i$ in the initial graph $G_0$. In particular, this construction provides a circuit of depth $2d$ and size $n^{k^{1/d} + O(1)}$ that computes $\STCONN(k(n))$.
  
Summarizing this discussion, it follows that for all $k \leq n$ and $\smash{d \leq \log k}$, $\STCONN(k(n))$ can be computed by depth-$2d$ circuits of size $\smash{n^{O(k^{1/d})}}$, or equivalently by depth-$d$ circuits of size $\smash{n^{O(k^{2/d})}.}$

\medskip

\noindent {\bf Lower bounds.}  Furst, Saxe, and Sipser~\cite{FSS:84} were the first to show that $\STCONN \eqdef \STCONN(n) \notin \acz$ via a reduction from their lower bound against small-depth circuits computing the parity function. By the same reduction, H{\aa}stad's subsequent optimal lower bound against parity~\cite{Hastad:86} implies that depth-$d$ circuits computing $\STCONN(k(n))$ must have size $2^{\Omega(k^{1/(d+1)})}$;\ignore{\inote{It seems to me that it follows from the FSS reduction that depth-$d$ circuits with bottom OR gates must have size $2^{\Omega(k^{1/d})}$, while depth-$d$ circuits with bottom AND gates must have size $2^{\Omega(k^{1/(d+1)})}$ (the reduction as described in FSS is performed by OR-AND subcircuits, thus there is a collapse of layers only in one of the aforementioned cases).}} in particular, for $k(n) = (\log n)^{\omega(1)}$ polynomial-size circuits computing $\STCONN(k(n))$ must have depth $d = \Omega(\log k/\log\log n)$.  Note, however, that this is not a useful bound for \emph{small} distance connectivity, since when $k(n) = o(\log n)$ the $\smash{2^{\Omega(k^{1/(d+1)})}}$ lower bound is less than $n$ and hence trivial.

Ajtai~\cite{Ajtai:89} was the first to show that $\STCONN(k(n)) \notin \acz$ for all $k(n)=\omega_n(1)$; however, his proof did not yield an explicit circuit size lower bound. His approach was further analyzed and simplified by
Bellantoni, Pitassi, and Urquhart~\cite{BPU:92}, who showed that this technique gives a (barely super-polynomial) 
$n^{\Omega(\log^{(d+3)} k)}$ lower bound on the size of depth-$d$ circuits for $\STCONN(k(n))$, where $\log^{(i)}$ denotes the $i$-times iterated logarithm. This implies that polynomial-size circuits computing $\STCONN(k(n))$ must have depth $\Omega(\log^* k)$.

Beame, Impagliazzo, and Pitassi~\cite{BIP:98} gave a significant quantitative strengthening of Ajtai's result in the regime where $k(n)$ is not too large.  For $k(n) \leq \log n$, they showed that any depth-$d$ circuit for $\STCONN(k(n))$ must have size $n^{\Omega(k^{\phi^{-2d}/3})},$ where $\phi = (\sqrt{5}+1)/2$ is the golden ratio.  Their arguments are based on a special-purpose ``connectivity switching lemma'' that they develop, which combines elements of both the Ajtai \cite{ajtai1983} ``independent set
style'' switching lemma and the later approach to switching lemmas given by Yao \cite{Yao:85}, Hastad \cite{Hastad:86} and Cai \cite{Cai86}.

Observe that the \cite{BIP:98} lower bound shows that polynomial-size circuits for 
$\STCONN(k(n))$ require depth $\Omega(\log \log k)$ (and as noted above, the \cite{BIP:98} lower bound only holds for $k(n) \leq \log n$).  Beame et al.~asked whether this $\Omega(\log \log k)$ could be improved to
$\Omega(\log k)$, which is optimal by the upper bound sketched above.  This was achieved recently by Rossman \cite{Rossman:13}, who showed that for $k(n) \leq \log \log n$, polynomial-size circuits for $\STCONN(k(n))$ require depth
$\Omega(\log k)$.  In more detail, he showed
that for $k(n) \leq \log \log n$ and $d(n) \le \log n/(\log \log n)^{O(1)}$,  depth-$d$ \emph{formulas} for $\STCONN(k(n))$ require size $n^{\Omega(\log k)}$.  By the trivial relation between formulas and circuits (every circuit of size $S$ and depth $d$ is computed by a formula of size $S^d$ and depth $d$), this implies that 
for such $k(n)$ and $d(n)$, depth-$d$ \emph{circuits} for $\STCONN(k(n))$ require size $n^{\Omega((\log k)/d)}.$
While this answers the question of Beame et al., the 
$n^{\Omega((\log k)/d)}$ circuit size bound that follows from Rossman's formula size bound  is significantly smaller than the $n^{\Omega(k^{\phi^{-2d}/3})}$ circuit size bound of~\cite{BIP:98} when $d$ is small.  Furthermore, Rossman's result only holds for $k(n) \le \log \log n$ whereas \cite{BIP:98}'s holds for $k(n) \le \log n$ (and ideally we would like a lower bound for all distances $k(n) \le n$).
%\ignore{(It should be noted that the main contributions of \cite{Rossman:13} are in separating the computational power of small-depth formulas and small-depth circuits and are orthogonal to the results mentioned above.)}
\subsection{Our results}\label{sec:subsection_our_resuts}

\begin{table}[t] 
\renewcommand{\arraystretch}{1.6}
\centering
\begin{tabular}{|c|c|c|c|}
\hline
 & Circuit size & Depth of poly-size circuits & Range of $k$'s  \\ \hline
  Implicit in~\cite{Hastad:86} &  $2^{\Omega(k^{1/(d+1)})}$  & $\Omega(\log k / \log \log n)$ & 
All $k$\\ \hline 
\cite{Ajtai:89,BPU:92}  &\ \  $n^{\Omega(\log^{(d+3)} k)}$ \ \ & $\Omega(\log^* k)$ & All $k$ \\ \hline
\cite{BIP:98} & $n^{\Omega(k^{\phi^{-2d}/3})}$ & $\Omega(\log \log k)$ & $k \le \log n$ \\ \hline
\cite{Rossman:13} &  $n^{\Omega((\log k) /d)}$  & $\Omega(\log k)$ & 
$k \le \log \log n$ \\ \hline \hline
\ \ Folklore upper bound\ \  & $n^{O(k^{2/d})}$ & $2\log k$ & All $k$ \\ \hline
\multirow{2}{*}{{\bf This work}} & $n^{\Omega(k^{1/d}/d)}$ & \multirow{2}{*}{$\Omega(\log k / \log \log k)$} & $k \leq n^{1/5}$ 
\vspace*{-2pt} \\
& $n^{\Omega(k^{1/5d}/d)}$ &  & All $k$ \\ \hline
\end{tabular}\vspace{0.2cm}\caption{Previous work and our results on the size of depth-$d$ circuits for $\STCONN(k(n))$.  The column ``Range of $k$'s'' indicates the values of $k$ for which the lower bound is proved to hold.}
\label{table:previous-work}
\end{table}

Our main result is a near-optimal lower bound for the small-depth circuit size of $\STCONN(k(n))$ for all distances $k(n) \le n$. We prove the following:

\begin{theorem} \label{thm:main} \hspace{-0.03cm}For any {$k(n) \leq n^{1/5}$} and any $d = d(n)$, any depth-$d$ circuit computing $\STCONN(k(n))$ must have size $n^{\Omega(k^{1/d}/ d)}$. Furthermore, for any $k(n) \le n$ and any $d = d(n)$, any depth-$d$ circuit computing $\STCONN(k(n))$ must have size $n^{\Omega(k^{1/5d}/d)}$.  
 \end{theorem}

Our lower bound is very close to the best possible, given the $n^{O(k^{2/d})}$ upper bound. Indeed, strengthening our theorem to $n^{k^{\Omega(1/d)}}$ for all values of $k$ and $d$ would imply a breakthrough in circuit complexity, showing that unbounded fan-in circuits of depth $o(\log n)$ computing $\STCONN$ must have super-polynomial size.  Since every function in $\mathsf{NC^1}$ can be computed by unbounded fan-in circuits of polynomial size and depth $O(\log n / \log \log n)$ (see e.g.~\cite{KPNY84}), such a strengthening would yield an unconditional proof that $\STCONN \notin \mathsf{NC^1}$.

Comparing to previous work, our $n^{\Omega(k^{1/d}/d)}$ lower bound subsumes the main  $n^{\Omega(k^{\phi^{-2d}/3})}$ lower bound result of Beame et al.~\cite{BIP:98} for all depths $d$, and improves the $n^{\Omega((\log k)/d)}$ circuit size lower bound that follows from Rossman's formula size lower bound~\cite{Rossman:13} except when $d$ is quite close to $\log k$ (specifically, except when $\Omega(\log k/\log\log k) \le d \le O(\log k)$). 
For large distances $k(n)$ for which the results of~\cite{BIP:98,Rossman:13} do not apply (i.e.~$k(n) = \omega(\log n)$), our lower bound subsumes the $2^{\Omega(k^{1/(d+1)})}$ lower bound that is implied by~\cite{Hastad:86} for all distances $k(n) \le n^{1/5}$ and depths $d$, and it subsumes the subsequent $n^{\Omega(\log^{(d+3)} k)}$ lower bound of~\cite{Ajtai:89,BPU:92} for all distances $k$ and depths $d$.

Another perspective on 
Theorem \ref{thm:main} is that it implies that polynomial-size circuits require depth $\Omega(\log k/\log \log k)$ 
to compute $\STCONN(k(n))$ for all distances $k(n) \le n$.  While Rossman's results give $\Omega(\log k)$, they hold only for the significantly restricted range $k(n) \leq \log \log n$. (And indeed, as noted above a lower bound of $\Omega(\log k)$ for all $k(n)$ would imply that $\STCONN \notin \mathsf{NC^1}$.)

\subsection{Our approach}

Previous state-of-the-art results on this problem employed rather sophisticated arguments and involved machinery.
Beame et al.~\cite{BIP:98} (as well as the earlier works of~\cite{Ajtai:89,BPU:92}) obtained their lower bounds by considering the $\STCONN(k(n))$ problem on layered graphs of permutations, i.e.,~graphs with $k + 1$ layers of $n$ vertices per layer in which the induced graph between adjacent layers is a perfect bipartite matching. They developed a special-purpose ``connectivity switching lemma'' that bounds the depth of specialized decision trees for randomly-restricted layered graphs.  
\ifnum\verbose=1
A node of such a decision tree queries a (vertex, direction) pair where ``vertex'' is a vertex in the graph and ``direction'' is either ``forward'' or ``backward.''  
\fi
Rossman~\cite{Rossman:13} considered random subgraphs of the ``complete $k$-layered graph'' (with $k+1$ layers of $n$ vertices and $kn^2$ edges) where each edge is independently present with probability $1/n$.  At the heart of his proof is an intricate notion of ``pathset complexity,'' which roughly speaking measures the minimum cost of constructing a set of paths via the operations of union and relational join, subject to certain ``density constraints.'' 
\ifnum\verbose=1
Half of Rossman's proof relates the small-depth formula complexity of $\STCONN(k(n))$ to pathset complexity, and the other half is an involved combinatorial lower bound on pathset complexity.
\fi
%\lnote{The vagueness of ``under a few density constraints'' is a little distracting here (e.g. what is a density constraint?), and I don't think it is possible to convey this clearly in a short phrase or even paragraph (and I don't think we want or have to). \red{Igor:} Yeahh we shouldn't get too technical here but if we decide to keep this ``roughly speaking...'' about pathset complexity then we need to mention the density constraints (which adds only a few words to the sentence). As Ben explains in the paper, without the density constraints there are pretty efficient constructions of dense pathsets from unions and relational joins. \red{Rocco:}  How about this kind of \violet{subject to certain ``density constraints.''} as a compromise?} 

In contrast, we feel that our approach is both conceptually and technically simple.  Instead of working with layered permutation graphs or random subgraphs of the complete layered graph, we consider a class of series-parallel graphs that are obtained in a straightforward way (see Section \ref{sec:reduction}) from a skewed variant of the ``Sipser functions'' that have played an important role in the classical circuit lower bounds of Sipser~\cite{sipser1983}, Yao~\cite{Yao:85}, and H{\aa}stad~\cite{Hastad:86}. Briefly, for every $d \in \N$, the $d$-th Sipser function $\Sipser_d$ is a read-once monotone formula with $d$ alternating layers of $\AND$ and $\OR$ gates of fan-in $w$, where $w\in \N$ is an asymptotic parameter that tends to $\infty$ (and so $\Sipser_d$ computes an $n = w^{d}$ variable function). Building on the work of Sipser and Yao, H{\aa}stad used the Sipser functions\footnote{The exact definition of the function used in \cite{Hastad:86} differs slightly from our description for some technical reasons which are not important here.}  to prove an optimal depth-hierarchy theorem for circuits, showing that for every $d \in \N$, any depth-$d$ circuit computing $\Sipser_{d+1}$ must have size $\exp(n^{\Omega(1/d)})$.

\ignore{In this work we use a }
The skewed variant of the Sipser functions that we use to prove our near-optimal lower bounds for $\STCONN(k(n))$
{is as follows}. For every $d \in \N$ and $2 \le u \le w$, the \emph{$d$-th $u$-skewed Sipser function}, denoted $\SkewedSipser_{u,d}$, is essentially $\Sipser_{2d+1}$ but with the $\AND$ gates having fan-in $u$ rather than $w$ (see Section~\ref{sec:reduction} for a precise definition; as we will see, the number of levels of $\AND$ gates is the key parameter for $\SkewedSipser$, which is why we write $\SkewedSipser_{u,d}$ to denote the $n$-variable formula that has $d$ levels of $\AND$ gates and $d+1$ levels of $\OR$ gates.)  Via a simple reduction given in Section \ref{sec:reduction}, we show that to get lower bounds for depth-$d$ circuits computing $\STCONN(u^d)$ on $n$-node graphs, it suffices to prove that depth-$d$ circuits for $\SkewedSipser_{u,d}$ must have large size. Under this reduction the fan-in of the $\AND$ gates is directly related to the length of (potential) paths between $s$ and $t$. This is why we must use a \emph{skewed} variant of the Sipser function in order to obtain lower bounds for small distance connectivity. We remark that even the case $u=2$ is interesting and can be used to get the $n^{\Omega(k^{1/d}/d)}$ lower bound of Theorem~\ref{thm:main} for $k$ up to roughly $2^{\sqrt{\log n}}$.  Allowing a range of values for $u$ enables us to get the lower bound for $k$ up to $n^{1/5}$ (as stated in Theorem~\ref{thm:main}).

Our main technical result of the paper is a lower bound for $\SkewedSipser_{u,d}$, \red{a formula of depth $2d + 1$ over $n = n(u,w,d) = u^d w^{d + 33/100}$ variables (for technical reasons we use a smaller fan-in for the first layer of OR gates next to the inputs).

\begin{theorem}
\label{thm:our-depth-hierarchy}
Let $d(w) \geq 1$ and $2 \leq u(w) \leq w^{33/100}$, where $w \to \infty$. Then any depth-$d$ circuit computing $\SkewedSipser_{u,d}$ has size at least $w^{\Omega(u)} = n^{\Omega(u/d)}$.%\inote{Justification: from $2 \leq u \leq w^{33/100}$ and $d\geq 1$ we get $w^{d} \leq n \leq w^{3d}$ thus $n = w^{\Theta(d)}$ as $w \to \infty$.}  
\end{theorem} 
}

Observe that setting $u = k^{1/d}$ this size lower bound is $n^{\Omega(k^{1/d}/d)}$, and therefore we indeed obtain the lower bound for $\STCONN(k(n))$ stated in Theorem~\ref{thm:main} as a corollary. As we point out in Section \ref{sec:proof_depth_hierarchy} \orange{(Remark \ref{remark:upperbound})},  the lower bound given in Theorem \ref{thm:our-depth-hierarchy} for $\SkewedSipser$ is essentially optimal.

Though they are superficially similar, Theorem~\ref{thm:our-depth-hierarchy} and H{\aa}stad's depth hierarchy theorem differ in two important respects.  Both result from our goal of using Theorem~\ref{thm:our-depth-hierarchy} to get lower bounds for small distance connectivity, and both pose significant challenges in extending H{\aa}stad's proof:

  \begin{enumerate}
  \item   H{\aa}stad showed that depth-$d$ unbounded fan-in circuits require large size to compute a single highly symmetric ``hard function,'' namely $\Sipser_{d+1}$. In contrast, toward our goal of understanding the depth-$d$ circuit size of $\STCONN(k(n))$ for \emph{all values of $k = k(n)$ and $d = d(n)$}, we seek lower bounds on the size of depth-$d$ unbounded fan-in circuits computing any one of a spectrum of asymmetric hard functions, namely $\SkewedSipser_{u,d}$ for all $u := k^{1/d}$ (with stronger quantitative bounds as $k$ and $u$ get larger). 
  \ignore{Furthermore, whereas H{\aa}stad's formula $\Sipser_{d+1}$ is highly regular in the sense that all its gates have fan-in $w$, we have to deal with highly asymmetric $\SkewedSipser_{u,d}$ formulas where $u$ may be much less than $w$. }
 
  \item To get the strongest possible result in his depth hierarchy theorem, H{\aa}stad (like Yao and Sipser) was primarily focused on lower bounding the size of circuits of depth exactly one less than $\Sipser_{d+1}$. In contrast, since in our framework our goal is to lower bound the size of depth-$d$ circuits computing $\SkewedSipser_{u,d}$ (corresponding to $\STCONN(k(n))$ with $k=u^d$) which has depth $2d+1$, we are interested in the size of circuits of depth (roughly) \emph{half} that of our hard function $\SkewedSipser_{u,d}$.
\end{enumerate} 

In Section \ref{sec:AAA} we recall the high-level structure of H{\aa}stad's proof of his depth hierarchy theorem (based on the method of random restrictions), highlight the issues that arise due to each of the two differences above, and  describe how our techniques --- specifically, the method of \emph{random projections} --- allow us to prove Theorem \ref{thm:our-depth-hierarchy} in a clean and simple manner. 

\section{H{\aa}stad's depth hierarchy theorem, random projections, and proof outline of Theorem \ref{thm:our-depth-hierarchy}} 
\label{sec:AAA}

{\bf H{\aa}stad's depth hierarchy theorem and its proof.}
Recall that H{\aa}stad's depth hierarchy theorem shows that $\Sipser_{d+1}$ cannot be computed by any circuit $C$ of depth $d$ and size $\exp(n^{O(1/d)}).$ The main idea is to design a sequence of random restrictions $\{ \calR_\ell\}_{2\le \ell\le d}$ satisfying two competing requirements:

\begin{itemize}

\item {\bf Circuit $C$ collapses.} The randomly restricted circuit  $C\uhr \brho^{(d)} \cdots \brho^{(2)}$, where $\brho^{(\ell)}\leftarrow \calR_\ell$ for $2 \le \ell \le d$, collapses to a ``simple function'' with high probability. This is shown via iterative applications of a switching lemma for the $\calR_\ell$'s, where each application shows that with high probability a random restriction $\brho^{(\ell)}$ decreases the depth of the circuit $C \uhr \brho^{(d)}\cdots \brho^{(\ell+1)}$ by at least one.  The upshot is that while $C$ is a size-$S$ depth-$d$ circuit, $C\uhr \brho^{(d)} \cdots \brho^{(2)}$ collapses to a small-depth decision tree (i.e.~a ``simple function'') with high probability. 

\item {\bf Hard function $\Sipser_{d+1}$ retains structure.}  In contrast with the circuit $C$, the hard function $\Sipser_{d+1}$ is ``resilient'' against the random restrictions $\brho^{(\ell)}\leftarrow\calR_\ell$.  In particular, each random restriction $\brho^{(\ell)}$ simplifies $\Sipser$ only by one layer, and so $\Sipser_{d+1} \uhr \brho^{(d)} \cdots \brho^{(\ell)}$ contains $\Sipser_\ell$ as a subfunction with high probability. Therefore, with high probability. $\Sipser_{d+1} \uhr \brho^{(d)} \cdots \brho^{(2)}$ still contains $\Sipser_2$ as a subfunction, and hence is a ``well-structured function'' which cannot be computed by a small-depth decision tree.
\end{itemize}
We remind the reader that to satisfy these competing demands, the random restrictions $\{\calR_\ell\}$ devised by H{\aa}stad specifically for his depth hierarchy theorem are not the ``usual'' random restrictions where each coordinate is independently kept alive with probability $p \in (0,1)$, and set to a uniform bit otherwise (it is not hard to see that $\Sipser$ does not retain structure under these random restrictions). Likewise, the switching lemma for the $\calR_\ell$'s is not the ``standard'' switching lemma (which H{\aa}stad used to obtain his optimal lower bounds against the parity function). Instead, at the heart of H{\aa}stad's proof are new random restrictions $\{ \calR_\ell\}_{2\le \ell\le d}$ designed to satisfy both requirements above: the coordinates of $\calR_\ell$ are carefully correlated so that $\Sipser_{\ell+1}$ retains structure, and H{\aa}stad proved a special-purpose switching lemma showing that $C$ collapses under these carefully tailored new random restrictions. 
  
\medskip

{\bf Issues that arise in our setting.}  At a technical level (related to point (1) described at the 
  end of Section \ref{sec:intro}),
H{\aa}stad's special-purpose switching lemma is not useful for analyzing our $\SkewedSipser_{u,d}$ formulas for most values of $u = k^{1/d}$ of interest, since they have a ``fine structure'' that is destroyed by his too-powerful random restrictions.  His switching lemma establishes that any DNF of width $n^{O(1/d)}$ collapses to a small-depth decision tree with high probability when it is hit by a random restriction $\brho^{(\ell)} \leftarrow \calR_\ell$.  
Observe that his hard function $\Sipser_{d+1}$ has DNF-width $\Omega(\sqrt{n})$, so his switching lemma does not apply to it (and indeed as discussed above, hitting $\Sipser_{d+1}$ with his random restriction results in a well-structured function that still contains $\Sipser_{d}$ as a subfunction with high probability).  In contrast, in our setting the hard function $\SkewedSipser_{u,d}$ has $d$ levels of $\AND$ gates of fan-in $u$, and in particular, can be written as a DNF of width $u^d = k$. So for all $k = k(n)$ and $d = d(n)$ such that $k \ll n^{O(1/d)}$ (indeed, this holds for most values of $k$ and $d$ of interest), the relevant hard function $\SkewedSipser_{u,d}$ collapses to a small-depth decision tree after a single application of H{\aa}stad's random restriction. 

Next (related to point (2)), recall that the formula computing H{\aa}stad's hard function $\Sipser_{d+1}$ has a highly regular structure where the fan-ins of all gates --- both $\AND$'s and $\OR$'s --- are the same. As discussed above, H{\aa}stad employs a random restriction which (with high probability) ``peels off'' a single layer of $\Sipser_{d+1}$ and results in a function that contains $\Sipser_d$ as a subfunction. Due to their regular structures, $\Sipser_d$ is dual to $\Sipser_{d+1}$ (more precisely, the bottom-layer depth-$2$ {subcircuits} of $\Sipser_{d}$ are dual to those of $\Sipser_{d+1}$), and this allows H{\aa}stad to repeat the same procedure $d-1$ times. In contrast, in our setting we are dealing with the highly asymmetric $\SkewedSipser_{u,d}$ formulas where the fan-ins of the $\AND$ gates are much less than those of the $\OR$ gates. Therefore, in order to reduce to a smaller instance of the same problem, our setup requires that we peel off two layers of $\SkewedSipser_{u,d}$ at a time rather than just one as in H{\aa}stad's argument. To put it another way, while H{\aa}stad's switching lemma uses a single layer of his hard function $\Sipser_{d+1}$ (i.e.~disjoint copies of $\OR$'s/$\AND$'s of fan-in $w$) to ``trade for'' one layer of depth reduction in $C$, our switching lemma will use two layers of our hard function $\SkewedSipser_{u,d}$ (i.e.~disjoint copies of read-once CNF's with $\smash{u = k^{1/d}}$ clauses of width $w$) to trade for one layer of depth reduction in $C$.

\medskip
{\bf Our approach:  random projections.}
A key technical ingredient in H{\aa}stad's proof of his depth hierarchy theorem --- and indeed, in the works of~\cite{BIP:98,Rossman:13} on $\STCONN(k(n))$ as well --- is the \emph{method of random restrictions}.  In particular, they all employ \emph{switching lemmas} which show that a randomly-restricted small-width DNF collapses to a small-depth decision tree with high probability: as mentioned above, H{\aa}stad proved a special-purpose switching lemma for random restrictions tailored for the Sipser functions, while Beame et al.~developed a ``connectivity switching lemma'' for random restrictions of layered permutation graphs, and Rossman used H{\aa}stad's ``usual'' switching lemma in conjunction with his pathset complexity machinery. 

In this paper we work with \emph{random projections}, a generalization of random restrictions. Given a set of formal variables $\calX = \{x_1, . . . , x_n\}$, a restriction $\brho$ either fixes a variable $x_i$ (i.e.~$\brho(x_i) \in \{0,1\}$) or keeps it alive (i.e.~$\brho(x_i) = x_i$, often denoted by $\ast$). A \emph{projection}, on the other hand, either fixes $x_i$ or maps it to a variable $y_j$ from a possibly different space of formal variables $\calY = \{y_1, . . . , y_m\}$. Restrictions are therefore a special case of projections where $\calY \equiv \calX$, and each $x_i$ can only be fixed or mapped to itself. (See Section \ref{sec:random_projection} for precise definitions.) Our arguments crucially employ projections in which $\calY$ is smaller than $\calX$, and where moreover each $x_i$ is only mapped to a specific element $y_j$ where $j$ depends on $i$ in a carefully designed way that depends on the structure of the formula computing the $\SkewedSipser$ function. Such ``collisions'', where %blocks of distinct 
  multiple formal variables in $\calX$ are mapped to the same new formal variable $y_j \in \calY$, play an important role in our approach.

Random projections were used in the recent work of Rossman, Servedio, and Tan~\cite{RST:15}, where they are the key ingredient enabling that paper's average-case extension of H{\aa}stad's worst-case depth hierarchy theorem. In earlier work, Impagliazzo, Paturi, and Saks~\cite{IPS:97} used random projections to obtain size-depth tradeoffs for threshold circuits, and Impagliazzo and Segerlind~\cite{IS01} used them to establish lower bounds against constant-depth Frege systems in proof {complexit\-y}. Our work  provides further evidence for the usefulness of random projections in obtaining strong lower bounds: random projections allow us to obtain sharper quantitative bounds while employing simpler arguments, both conceptually and technically, than in the previous works \cite{Ajtai:89,BPU:92,BIP:98,Rossman:13} on the small-depth complexity of $\STCONN(k(n))$.

We remark that although~\cite{RST:15} and this work both employ random projections to reason about the Sipser function (and its skewed variants), the main advantage offered by projections over restrictions are different in the two proofs. In~\cite{RST:15} the overarching challenge was to establish \emph{average-case} hardness, and the identification of variables was key to obtaining uniform-distribution correlation bounds from the composition of highly-correlated random projections. As outlined above, in this work a significant challenge stems from our goal of understanding the depth-$d$ circuit size of $\STCONN(k(n))$ for \emph{all values of $k = k(n)$ and $d = d(n)$}. The added expressiveness of random projections over random restrictions is exploited both in the proof of our projection switching lemma (see Section~\ref{sec:proof-outline} below) and in the arguments establishing that our $\SkewedSipser_{u,d}$ functions ``retain structure" under our random projections.

\subsection{Proof outline of Theorem \ref{thm:our-depth-hierarchy}} 
\label{sec:proof-outline}
Our approach shares the same high-level structure as H{\aa}stad's depth hierarchy theorem, and is based on a sequence $\boldsymbol{\Psi}$ of $d-1$ random projections satisfying two competing requirements (it will be more natural for us to present them in the opposite order from our discussion of H{\aa}stad's theorem in the previous section): 

\begin{itemize}

\item {\bf Hard function $\SkewedSipser$ retains structure.} Our random projections are defined with the hard function $\SkewedSipser$ in mind, and are carefully designed so as to ensure that $\SkewedSipser_{u,d}$ ``retains structure'' with high probability under their composition $\mathbf{\Psi}$. 

In more detail, each of the $d-1$ individual random projections comprising $\mathbf{\Psi}$ peels off two layers of $\SkewedSipser$, and a randomly projected $\SkewedSipser_{u,\ell}$ contains $\SkewedSipser_{u,\ell-1}$ as a subfunction with high probability. These individual random projections are simple to describe: each bottom-layer depth-$2$ subcircuit of $\SkewedSipser_{u,\ell}$ (a read-once CNF with 
 $\smash{u = k^{1/d}}$ clauses of width $w$) independently ``survives'' with probability $q\in (0,1)$ and is ``killed'' with probability $1-q$ (where $q$ is a parameter of 
   the restrictions), and 
\begin{itemize}
\item if it survives, all $uw$ variables in the CNF are {\bf projected} to the same fresh formal variable  (with different CNFs mapped to different formal variables); 
\item if it is killed, all its variables are {\bf fixed} according to a random $0$-assignment of the CNF chosen uniformly from a particular set of $2u$ many $0$-assignments.
\end{itemize} 
In other words, each bottom-layer depth-$2$ subcircuit independently simplifies to a fresh formal variable (with probability $q$) or the constant $0$ (with probability $1-q$). With the appropriate definition of $\SkewedSipser$ and choice of $q$, it is easy to verify that indeed a randomly projected $\SkewedSipser_{u,\ell}$ contains $\SkewedSipser_{u,\ell-1}$ as a subfunction with high probability.  (For this to happen, the fanin of the bottom OR gates of $\SkewedSipser$
  is chosen to be moderately smaller than $w$, the fanin of all other OR gates in $\SkewedSipser$;
  see Definition \ref{Sipser:def} for details.)

\item {\bf Circuit $C$ collapses.}  In contrast with $\SkewedSipser_{u,d}$, any depth-$d$ circuit $C$ of size $n^{O(u/d)}$ collapses to a small-depth decision tree under $\boldsymbol{\Psi}$ with high probability. Following the standard ``bottom-up'' approach to proving lower bounds against small-depth circuits, we establish this by arguing that each of the individual random projections comprising $\mathbf{\Psi}$ ``contributes to the simplification'' of $C$ by reducing its depth by (at least) one.  

More precisely, in Section \ref{sec:projection_switching_lemma} we prove a \emph{projection switching lemma}, showing that a small-width DNF or CNF ``switches'' to a small-depth decision tree with high probability under our random projections. (The depth reduction of $C$ follows by applying this lemma to every one of its bottom-level depth-$2$ subcircuits.)  Recall that the random projection of a depth-$2$ circuit over a set of formal variables $\calX$ yields a function over a new set of formal variables $\calY$, and in our case $\calY$ is significantly smaller than $\calX$.  In addition to the structural simplification that results from setting variables to constants (as in the switching lemmas of~\cite{Hastad:86,BIP:98,Rossman:13} for random \emph{restrictions}), the proof of our projection switching lemma also exploits the additional structural simplification that results from distinct variables in $\calX$ being mapped to the same variable in $\calY$. 
\end{itemize}

% !TEX root =  main.tex

{\subsection{Preliminaries} \label{sec:preliminaries}

%We briefly recall some standard terminology and notation.

A \emph{restriction} over a finite set of variables $A$ is an element of $\{0,1,\ast\}^A.$  We define the \emph{composition} $\rho \rho'$ of two restrictions $\rho, \rho' \in \{0,1,\ast\}^{A}$ over a set of variables $A$ to be the restriction
\[
(\rho\rho')_\alpha \; \eqdef \; 
\begin{cases}
\rho_\alpha &\text{if\ } \rho_\alpha  \neq \ast\\
\rho'_\alpha& \text{otherwise}
\end{cases},
\quad \text{for all $\alpha \in A$.}
\]

A \emph{DNF} is an $\OR$ of $\AND$s (terms) and a \emph{CNF} is an $\AND$ of $\OR$s (clauses).  The \emph{width} of a DNF (respectively, CNF) is the maximum number of variables that occur in any one of its terms (respectively, clauses). 

The \emph{size} of a circuit is its number of gates, and the \emph{depth} of a circuit is the length of its longest root-to-leaf path.
We count input variables as gates of a circuit (so any circuit for a function~that depends on all $n$ input variables trivially has size
at least $n$).  We will assume throughout the paper that circuits are \emph{alternating}, meaning that every root-to-leaf path alternates between $\AND$ gates and $\OR$ gates.  We also assume that circuits are \emph{layered}, meaning that for every gate $\mathsf{G}$,~every~root-to-{\sf G} path has the same length.  
These assumptions are without loss of generality as
by a standard conversion (see e.g. the discussion at \cite{stackexchange-layering-circuits}),~every depth-$d$ size-$S$ circuit is equivalent to a depth-$d$ alternating layered circuit of size at most $\poly(S)$ (this polynomial
increase is offset by the ``$\Omega(\cdot)$'' notation in the exponent of all of our theorem statements.)
}

% !TEX root =  main.tex

%\newpage

\section{Lower bounds against $\SkewedSipser$ yield lower bounds for small distance connectivity} \label{sec:reduction}

In this section we define $\SkewedSipser_{u,\ourd}$ and show that computing this formula on a particular input $z$ is equivalent to solving small-distance connectivity on a certain undirected (multi)graph $G(z)$. In a bit more detail, every input $z$ corresponds to a subgraph $G(z)$ of a fixed ground graph $G$ that depends only on $\SkewedSipser_{u,\ourd}$. (Jumping ahead, we associate each input bit of $\SkewedSipser_{u,\ourd}$ with an edge of its corresponding ground graph $G$.) Roughly speaking, AND gates translate into sequential paths, while OR gates correspond to parallel paths. After defining $\SkewedSipser_{u,\ourd}$ and describing this reduction, we give the proof of Theorem \ref{thm:main}, assuming Theorem \ref{thm:our-depth-hierarchy}.

The $\SkewedSipser$ formula is defined in terms of an integer parameter $w$; in all our results this is an asymptotic parameter that approaches $+ \infty$, and so $w$ should be thought of as ``sufficiently large'' throughout the paper. 

\begin{definition}\label{Sipser:def}
For $2 \le u \le w$ and $\red{d\ge 0}$,  $\SkewedSipser_{u,\ourd}$ is the Boolean function computed by the following monotone read-once formula: 

\begin{itemize}
\item There are $2\ourd+1$ alternating layers of $\OR$ and $\AND$ gates, 
  where the top and bottom-layer gates are $\OR$ gates.  (So there are $\ourd+1$ layers of $\OR$ gates and $\ourd$ layers of $\AND$ gates.) 
\item $\AND$ gates all have fan-in $u$.
\item  $\OR$ gates all have fan-in $w$, except bottom-layer $\OR$ gates which have fan-in $\red{w^{33/100}}$; we assume that $w^{\red{1/100}}$ is an integer throughout the paper. (The most important thing about the constant 33/100 in the above definition is that it is less than 1; the particular value 33/100  was chosen for technical reasons so that we could get the constant 5 in Theorem~\ref{thm:main}.)
\end{itemize}
Consequently, $\SkewedSipser_{u,\ourd}$ is a Boolean function over %$N(\ourd) := 
$n = (uw)^\ourd  \red{w^{33/100}}$ variables in total.
\end{definition}

\medskip
\noindent \textbf{From $\SkewedSipser_{u,\ourd}$ to small-distance connectivity.}  There is a natural correspondence between read-once  monotone Boolean formulas and series-parallel multigraphs in which each graph has a special designated ``start'' node
$s$ and a special designated ``end'' node $t$.  We now describe this correspondence via the inductive structure of read-once monotone Boolean formulas.  As we shall see, under this correspondence there is a bijection between the \emph{variables} of a formula $f$ and the \emph{edges} of the graph $G(f)$.
\begin{itemize}

\item  If $f(x)=x$ is a single variable, then the graph $G(f)$ has vertex set $V(f) = \{s,t\}$ and edge set $E(f)$ consisting of a single edge $\{s,t\}.$  

\item Let $f_1,\dots,f_m$ be read-once monotone Boolean formulas over disjoint sets of variables, where $G(f_i)$ is the (multi)graph associated with $f_i$ and $s_i,t_i$ are the start and end nodes of $G(f_i$).  

\begin{itemize}

\item If $f=\AND(f_1,\dots,f_m)$:   The graph $G(f)$ is obtained by identifying $t_1$ with $s_2$, $t_2$ with $s_3$,
\dots, and $t_{m-1}$ with $s_m$.  The start node of $G(f)$ is $s_1$ and the end node is $t_m$.  Thus the vertex set 
$V(f)$ is $V(f_1) \cup \cdots \cup V(f_m) \setminus
\{t_1,\dots,t_{m-1}\}$ and  the edge set  $E(f)$ is the multiset $E'(f_1) \cup \cdots \cup E'(f_m)$, where each $E'(f_i)$ is obtained from $E(f_i)$ by renaming the appropriate vertices. 

\item If $f=\OR(f_1,\dots,f_m)$:  The graph $G(f)$ is obtained by identifying $s_1,\dots,s_m$ all to a new start vertex $s$ and $t_1,\dots,t_m$ all to a new end vertex $t$.  Thus the vertex set 
$V(f)$ is $V(f_1) \cup \cdots \cup V(f_m) \cup \{s,t\} \setminus
\{s_1,\dots,s_m,t_1,\dots,t_m\}$ and  the edge set  $E(f)$ is the multiset $E'(f_1) \cup \cdots \cup E'(f_m)$, where again each $E'(f_i)$ is obtained from the corresponding edge set $E(f_i)$ by renaming vertices accordingly.
\end{itemize}
\end{itemize}
Since $f$ is read-once, the number of edges of $G(f)$ is precisely the number of variables of $f$, and there is a natural correspondence between edges and  variables. Figure \ref{figure:1} provides a concrete example of this construction.
%\inote{Made the graph black-and-white since the Boolean circuit is also black-and-white.}
%See Figure \ref{figure:1} for an illustration of this correspondence.

\begin{figure}[t]
\begin{center}
\vskip -.25in
\hskip -3.25in
\includegraphics[scale=0.28, trim = 0cm 0cm 0cm 0cm]{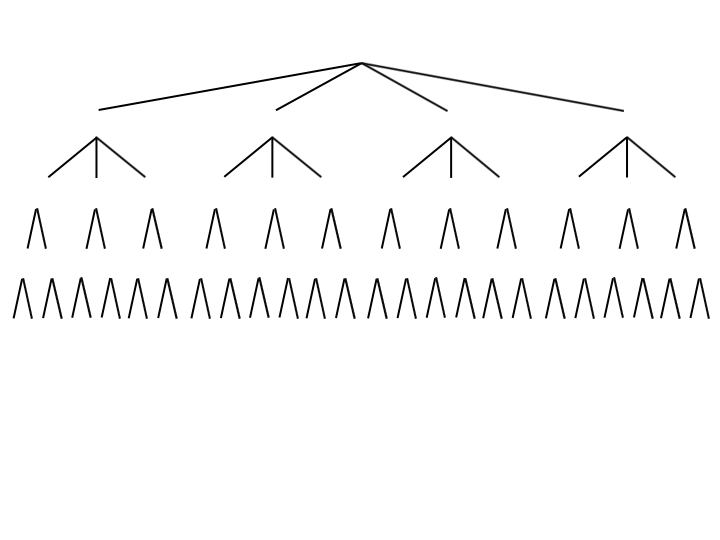} 
\vskip-2in
\hskip 3.25in
\includegraphics[height=1.25in, width=3in, trim = 0cm 0cm 0cm 0cm]{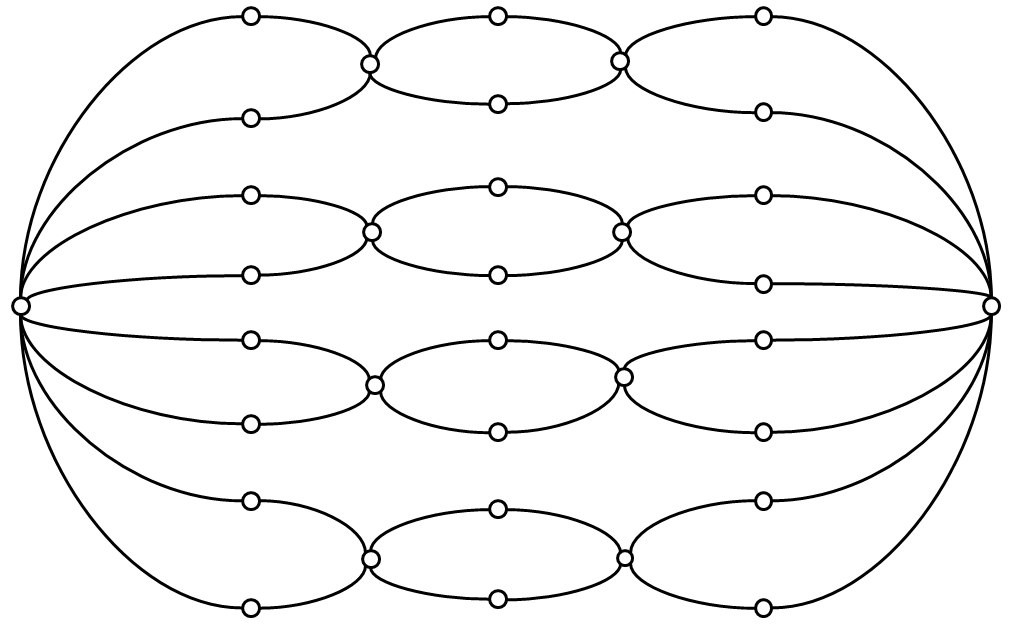}
%\includegraphics[width=4.5in,height=1.5in, trim = 0cm 0cm 0cm 0cm]{skewedsipserdagger}
%\vskip -0.93in  
%$s$~~~~~~~~~~~~~~~~~~~~~~~~~~~~~~~~~~~~~~~~~~~~~~~~~
%~~~~~~~~~~~~~~~~~~~~~~~~~~~~~~~~~~~~~~$t$
%\vskip 0.93in 
\vskip -0.7in 
\begin{picture}(400,50)(20, 40)
\put(101,123){{\tiny $\vee$}}

\put(26,102){{\tiny $\wedge$}}
\put(76,102){{\tiny $\wedge$}}
\put(125.5,102){{\tiny $\wedge$}}
\put(175.5,102){{\tiny $\wedge$}}

\put(9.5,83){{\tiny $\vee$}}
\put(26,83){{\tiny $\vee$}}
\put(42,83){{\tiny $\vee$}}
\put(59.5,83){{\tiny $\vee$}}
\put(76,83){{\tiny $\vee$}}
\put(92,83){{\tiny $\vee$}}
\put(109.25,83){{\tiny $\vee$}}
\put(125.5,83){{\tiny $\vee$}}
\put(141.5,83){{\tiny $\vee$}}
\put(159.5,83){{\tiny $\vee$}}
\put(175.5,83){{\tiny $\vee$}}
%\put(191,83){{\tiny $\vee$}}
\put(192,83){{\tiny $\vee$}}

\put(6,62.5){{\tiny $\wedge$}}
\put(13,62.5){{\tiny $\wedge$}}

\put(22.5,62.5){{\tiny $\wedge$}}
\put(29.5,62.5){{\tiny $\wedge$}}

\put(38.5,62.5){{\tiny $\wedge$}}
\put(45.5,62.5){{\tiny $\wedge$}}

\put(56,62.5){{\tiny $\wedge$}}
\put(63,62.5){{\tiny $\wedge$}}

\put(72.5,62.5){{\tiny $\wedge$}}
\put(79.5,62.5){{\tiny $\wedge$}}

\put(88.5,62.5){{\tiny $\wedge$}}
\put(95.5,62.5){{\tiny $\wedge$}}

\put(105.75,62.5){{\tiny $\wedge$}}
\put(112.75,62.5){{\tiny $\wedge$}}

\put(122,62.5){{\tiny $\wedge$}}
\put(129,62.5){{\tiny $\wedge$}}

\put(138,62.5){{\tiny $\wedge$}}
\put(145,62.5){{\tiny $\wedge$}}

\put(156,62.5){{\tiny $\wedge$}}
\put(163,62.5){{\tiny $\wedge$}}

\put(172,62.5){{\tiny $\wedge$}}
\put(179,62.5){{\tiny $\wedge$}}

\put(187.5,62.5){{\tiny $\wedge$}}
%\put(194.5,62.5){{\tiny $\wedge$}}
\put(196,62.5){{\tiny $\wedge$}}

\put(0,43){{\tiny $x_1 \cdots \cdots \cdots \cdots \cdots \cdots \cdots \cdots \cdots \cdots \cdots \cdots \cdots \cdots \cdots \cdots \cdots \cdots x_{48}$}}

\put (225,84){{$s$}}
\put (447,84){{$t$}}
\end{picture}
\vskip -.17in ~
\end{center}
%\vskip -1.75in
\caption{\small A read-once formula $f$ (on the left), which is a fan-in 4 OR of fan-in 3 $\AND$s of fan-in 2 $\OR$s of fan-in 2 
$\AND$s, and the corresponding graph $G(f)$ (on the right).}
\label{figure:1}
\end{figure}
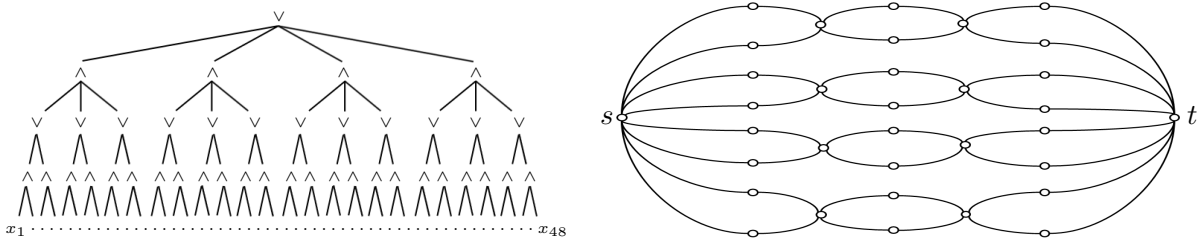

\begin{remark} \label{r:graph_and_vertices}
We note that if $f$ is \orange{a read-once monotone Boolean formula in which}
%an alternating formula in which
  \orange{the bottom-level gates are $\AND$ gates and have fan-in at least two,}
%each $\AND$ gate has fan-in at least two and the bottom-level gates are $\AND$ gates, 
  then $G(f)$ is a \emph{simple} graph and not a multigraph. 
%Since $G(f)$ is always connected, the number of vertices in $G(f)$ is at most twice its number of edges.\ignore{The number of nodes in $G(f)$ is at most twice the number of internal nodes of $f$ (it will be often much less, except for the base case of the construction described above).}
\end{remark}

A simple inductive argument gives the following:

\begin{observation} \label{obs:dr}
If $f$ is a read-once monotone alternating formula with $r$ layers of $\AND$ gates~of fan-ins $\alpha_1,\dots,\alpha_r$, respectively, then every shortest path from $s$ to $t$ in 
  the graph $G(f)$ has length exactly $\alpha_1 \cdots \alpha_r.$ Furthermore, if $H$ is a subgraph of $G(f)$ that contains some \orange{$s$-to-$t$} path, then it contains a path of length $\alpha_1 \cdots \alpha_r.$
\end{observation}

As a corollary, we have:
\begin{observation} \label{obs:sipser-base}
Every shortest path from $s$ to $t$ in $G(\SkewedSipser_{u,\ourd})$ has length exactly $u^{\ourd}.$
\end{observation}

Given a read-once monotone formula $f$ over variables $x_1,\dots,x_n$ and an assignment $z \in \{0,1\}^n$ to the
variables $x_1,\dots,x_n$, we define the graph $G(f,z)$ to be the (spanning) subgraph of $G(f)$ which has vertex
set $V(f,z) = V(f)$ and edge set $E(f,z)$ defined as follows:  each edge in $E(f)$ is present in $E(f,z)$ if and
only if the corresponding coordinate of $z$ is set to 1.  
A simple inductive argument gives the following:

\begin{observation} \label{obs:sipser-connectivity}
Given a read-once monotone alternating formula with $r$ layers of $\AND$ gates of fan-ins $\alpha_1,\dots,\alpha_r$, respectively, and an assignment $z\in \{0,1\}^n$, the graph $G(f,z)$ contains a path from $s$ to $t$ of length $\alpha_1 \cdots \alpha_r$ if and only if $f(z)=1$. 
\end{observation}

\ignore{

The next observation also follows from a simple inductive proof.

\begin{observation} \label{obs:paths_series_parallel}
Let $H$ be a spanning subgraph of $G(f)$, and assume that $s$ and $t$ are connected in $H$. Then there is a shortest path connecting $s$ and $t$ in $G(f)$ that is entirely contained in $H$.
\end{observation}
}

From these observations we obtain the following connection between $\SkewedSipser_{u,\ourd}$ and small-distance connectivity, which is key to our lower bound:\ignore{Observations \ref{obs:sipser-base}, \ref{obs:sipser-connectivity}, and \ref{obs:paths_series_parallel} imply the following characterization.}

\begin{corollary} \label{cor:sipser-connectivity}
The multigraph $G(\SkewedSipser_{u,\ourd},z)$ contains an $s$-to-$t$ path of length at most $u^{\ourd}$ if and only if $\SkewedSipser_{u,\ourd}(z)=1.$
\end{corollary}

Note that Corollary \ref{cor:sipser-connectivity} and Theorem \ref{thm:our-depth-hierarchy} together can be used to prove lower bounds for small-distance connectivity on \emph{multigraphs}. One way to obtain lower bounds for \emph{simple} graphs instead of multigraphs is by extending $\SkewedSipser_{u,\ourd}$ with an extra layer of fan-in two AND gates next to the input variables, then relying on Remark \ref{r:graph_and_vertices}. We use this simple observation and Theorem \ref{thm:our-depth-hierarchy} to establish Theorem \ref{thm:main}.

\begin{reptheorem}{thm:main}  
\hspace{-0.03cm}For any $k(n) \leq n^{1/\red{5}}$ and any $d = d(n)$, any depth-$d$ circuit computing $\STCONN(k(n))$ must have size $n^{\Omega(k^{1/d}/ d)}$. Furthermore, for any $k(n) \le n$ and any $d = d(n)$, any depth-$d$ circuit computing $\STCONN(k(n))$ must have size \violet{$n^{\Omega(k^{1/5d}/d)}$}. \end{reptheorem}
 
\begin{proof}
\blue{We assume that $d < {2 \log k/\log\log k}$ and $(k/2)^{1/d}\ge 2$
  (observe that the claimed bound is trivial if  $d \ge {2 \log k / \log \log k}$
  or $(k/2)^{1/d}<2$).}
\blue{Let $$u_0=\left\lfloor (k/2)^{1/d}\right\rfloor.$$
Then we have $u_0\ge 2$ and $u_0=\Omega(k^{1/d})$.}
% Let $u_0$ be the largest integer such that
%\begin{equation}\label{eq:def_u}
%2 \;\leq\; u_0^d \;\leq\; k/2.
%\end{equation}
%This integer is guaranteed to exist provided that  $d <\red{2 \log k/\log\log k}$  \red{and $k \geq C$ for
%a suitable absolute constant $C$} (we observe that the claimed bound is trivial if \red{$k < C$} or $d \ge %\red{2 \log k / \log \log k}$). 
%\blue{By our choice of $u_0$ we have $(u_0+1)^d>k/2$ and thus, $u_0=\Omega(k^{1/d})$.}
For convenience, let
\begin{equation}\label{eq:def_k0}
k_0 \eqdef u_0^d\le k/2 \quad \text{and} \quad n' \eqdef \left \lfloor \frac{n}{2} \right \rfloor.
\end{equation}
Further, let $w_0$ be the largest positive integer such that 
\begin{equation}\label{eq:def_w0}
(u_0 w_0)^d w_0^{\red{33/100}} \leq n'.
\end{equation}
Observe that, since $k\le n^{1/5}$ and $d<2\log k/\log \log k$, as $n\to +\infty$ we have
  similarly $w_0\to +\infty$.
\ignore{
Observe that $w_0 \geq 2$ since $k \leq n^{1/20}$ and $d < \log k$.}
\blue{Our choice of $w_0$ also implies that $w_0$ satisfies
$$u_0^d (w_0+1)^{d+33/100}>n'.$$ 
Let $n_0 \eqdef (u_0 w_0)^d w_0^{\red{33/100}}$.
Then from $d<2\log k/\log \log k$ and $w_0\rightarrow +\infty$ we have
$$
n_0\ge u_0^d\left(\frac{w_0+1}{2}\right)^{d+\frac{33}{100}}\!=\,\Omega(n/2^{d})=\omega(n^{0.99}) .%=\Omega(n^{4/5}).
$$
Combining this with $k\le n^{1/5}$ and $k_0\le k/2$ we have that 
  $k_0=o(n_0^{20/99})$ and $n_0\ge k_0^{4.9}$ when $n$ is sufficiently large.}
%It then follows from $k\le n^{1/5}$ and
%  $d \leq 2 \log k / \log \log k\le (2/5)\log n/\log \log (n^{1/5})$ that
%  $w_0\ge \log n$, and thus, as $n \to +\infty$ we have
%  similarly $w_0 \to +\infty.$ Moreover, letting
%\begin{equation}\label{eq:def_n0}
%n_0 \;\eqdef \; (u_0 w_0)^d w_0^{\red{33/100}},
%\end{equation}
%we have from $w_0\ge \log n$ and $d\leq 2\log k/\log \log k$ that 
%\begin{equation}\label{eq:out}
%n_0=k_0(w_0+1)^{d+33/100}\cdot \left(\frac{w_0}{w_0+1}\right)^{d+33/100}\ge \Omega(n).
%%k_0\cdot \Omega(n^{4/5})\cdot \frac{1}{(1+(1/w_0))^{d+33/100}}=\Omega(k_0\cdot n^{4/5}).
%\end{equation}}

We define a variant of our $\SkewedSipser_{u,d}$ formula so we can rely on Remark \ref{r:graph_and_vertices} and work directly with simple graphs instead of multigraphs. More precisely, let $\smash{\SkewedSipser^\dagger_{u_0,w_0,d}}$ be analogous to $\SkewedSipser_{u,d}$ with parameters $u_0$ (AND gate fan-in), $w_0$ (OR gate fan-in), and $d$ but containing an extra layer of fan-in 2 AND gates at the bottom connected to a new set of input variables.~In other words, this is a depth $2d + 2$ read-once
  alternating formula with twice the number of input variables of our original $\SkewedSipser$ formula (each input variable of $\SkewedSipser$ becomes an AND gate connected to two new fresh variables). Since $\SkewedSipser_{u_0, d}$ can be obtained by restricting $\smash{\SkewedSipser^\dagger_{u_0,w_0,d}}$ appropriately (i.e.~by setting to 1 a single variable in every new pair of variables) a lower bound on the circuit complexity of $\SkewedSipser$ immediately implies the same lower bound for 
    $\smash{\SkewedSipser^\dagger}$. 

In order to obtain a lower bound via Theorem \ref{thm:our-depth-hierarchy}, we need that $w_0^{\red{33/100}} \geq u_0$. This is equivalent to having $\smash{n_0 \geq u_0^{\red{133d/33 + \blue{1}\ignore{33/100}}}}$, which follows from $d\ge 2$
% Furthermore, by (\ref{eq:def_k0}) and $d \geq 2$ 
 \red{(we may
assume $d \geq 2$ since no~depth-1 circuit, i.e. single $\AND$ or $\OR$ gate, can compute $\STCONN(k(n))$)}
  and $n_0\ge k_0^{4.9}$ since 
%it is sufficient that 
  $$n_0 \geq  \red{k_0^\blue{4.9}}\ignore{21/5} > k_0^{\red{133/33 +\blue{1/2}. \ignore{33/200}}}$$ Consequently, we can apply Theorem \ref{thm:our-depth-hierarchy} to $\SkewedSipser_{u_0,d}$, and it follows from our discussion above that any depth-$d$ circuit computing $\smash{\SkewedSipser^\dagger_{u_0,w_0,d}}$ must have size at least
\begin{equation}\label{eq:size_lower_bound}
 n_0^{\Omega(u_0/d)} = n^{\Omega(k^{1/d}/d)}. 
\end{equation}
%provided that $\blue{n_0\ge k_0^{4.6}}$, 
 % \blue{which follows from (\ref{eq:out}), $k_0\le k/2$, and
  %$k\le n^{1/5}$}.%$k_0 \leq \red{n_0^\blue{4.6}\ignore{5/21}}$.

In the rest of the proof we translate (\ref{eq:size_lower_bound}) into a lower bound for $\STCONN(k(n))$. Following~the explanation given above, we consider the simple graph $G(\SkewedSipser^\dagger)$ with appropriate parameters. %(Figure \ref{figure:1} illustrates the construction of $G(\SkewedSipser^\dagger)$ from $\SkewedSipser^\dagger$.)\inote{Remove this sentence, since the picture no longer corresponds to our new family of functions with bottom fan-in $w^{33/100}$?} 
Since $u_0^d\le k/2$, it follows from 
  the same argument used to establish Corollary  \ref{cor:sipser-connectivity} that the graph $\smash{G(\SkewedSipser^\dagger_{u_0,w_0,\ourd},z)}$ contains an $s$-to-$t$ path of length at most $2u_0^{\ourd}\le k$ if and only~if~we have $\smash{\SkewedSipser^\dagger_{u_0,w_0,\ourd}(z)=1}$. %Moreover, Equation \ref{eq:def_u} and Observation \ref{obs:dr} imply that $G(\SkewedSipser^\dagger_{u_0,w_0,\ourd},z)$ contains an $s$-to-$t$ path of length at most $k$ precisely when $\SkewedSipser^\dagger_{u_0,w_0,\ourd}(z)=1$. Finally, 
Because $\smash{G(\SkewedSipser^\dagger_{u_0,w_0,\ourd})}$ has no isolated vertices and has~$n_0$ edges, it contains at most $2n_0\le 2n'\le n$ vertices by (\ref{eq:def_k0}) and (\ref{eq:def_w0}).
%, and (\ref{eq:def_n0}). 
Thus, a circuit $\mathcal{C}$ that computes $\STCONN(k(n))$ on undirected simple graphs on $n$ vertices can also be used to compute the formula $\smash{\SkewedSipser^\dagger_{u_0,w_0,\ourd}}$, 
%From the assumption that $k \leq n^{1/\red{5}}$ it follows that $k_0 \leq n_0^{\red{5/21}}$, 
and (\ref{eq:size_lower_bound}) yields that $\mathcal{C}$ must have size $n^{\Omega(k^{1/d}/d)}$. This completes the first part~of Theorem \ref{thm:main}. 

It remains to prove the lower bound for $\STCONN(\orange{k'}(n))$ with \orange{$n^{1/5} < k'(n) \leq n$. % < k'(n) \leq n$. 
For this,} let $\orange{k}(n) \eqdef n^{1/\red{5}}$.  We have established above that computing $\STCONN(k(n))$ on subgraphs of $\smash{G(\SkewedSipser^\dagger_{u_0,w_0,\ourd})}$ using depth-$d$ circuits requires size 
  $n^{\Omega(k^{1/d}/d)}$. However, a subgraph of $\smash{G(\SkewedSipser^\dagger_{u_0,w_0,\ourd})}$ contains an \orange{$s$-to-$t$} path of length at most $k(n)$ if and only if it contains a path from $s$ to $t$ of length at most $k'(n)$ (Observation \ref{obs:dr}). Consequently, any circuit $\mathcal{C}$ that computes $\STCONN(k'(n))$ on general $n$-vertex graphs can be used to compute $\STCONN(k(n))$ on subgraphs of $G(\SkewedSipser^\dagger_{u_0,w_0,\ourd})$ (by setting some input edges to 0). In particular, $\mathcal{C}$ must have size $$n^{\Omega(k^{1/d}/d)} = n^{\Omega(n^{1/5d}/d)} = n^{\Omega(\red{k'}^{1/5d}/d)}.$$
  This completes the second part of Theorem \ref{thm:main}.
% But since $\ell(n) = k(n)^{\Omega(1)}$, this is $n^{k^{\Omega(1/d)}/d}$, which completes the proof.
\end{proof}

\violet{
\begin{remark}
It is not hard to see that our~reduction in fact also captures other natural graph problems such as directed $k$-path (``Is there a directed path of length $k(n)$ in $G$?'') and directed $k$-cycle (``Is there a directed cycle of length $k(n)$ in $G$?''), and hence the lower bounds of Theorem~\ref{thm:main} apply to these problems as well.  This suggests the possibility of similarly obtaining other lower bounds from (variants of) depth hierarchy theorems for Boolean circuits, and we leave this as an avenue for further investigation.
\end{remark}
}

\def\ourd{d}
\def\d2sipser{\mathsf{Depth2Sipser}}
\def\d2sipser{\mathsf{CNFSipser}}

%\section{Our Target: Skewed Sipser}  \label{sec:target}
%
%Our $\SkewedSipser$ formula is defined in terms of an integer parameter $w$; in all our results this is an asymptotic parameter that approaches $+ \infty$, and so $w$ should be thought of as ``sufficiently large'' throughout the paper. 
%
%\begin{definition}
%For $2 \le u \le w$ and $\ourd \in \N$,  $\SkewedSipser_{u,\ourd}$ is the Boolean function computed by the following monotone read-once formula: 
%
%\begin{itemize}
%\item There are $2\ourd+1$ alternating layers of $\OR$ and $\AND$ gates, 
%  where the top and bottom-layer gates are $\OR$ gates.  (So there are $\ourd+1$ layers of $\OR$ gates and $\ourd$ layers of $\AND$ gates.) 
%\item $\AND$ gates all have fan-in $u$.
%\item  $\OR$ gates all have fan-in $w$, except bottom-layer $\OR$ gates which have fan-in $\sqrt{w}$.  
%\end{itemize}
%So $\SkewedSipser_{u,\ourd}$ is a Boolean function over %$N(\ourd) := 
%$(uw)^\ourd  \sqrt{w}$ variables in total.

\section{The Random Projection}\label{sec:random_projection}

In this section we define our random projections, which will be crucial in the proof of Theorem \ref{thm:our-depth-hierarchy}. First, we introduce 
  notation to manipulate the first two layers of $\SkewedSipser_{u,\ourd}$.

\begin{definition}
For $2 \le u \le w$, we define $\d2sipser_{u}$ to be the Boolean function computed by the following monotone read-once  formula: 
\begin{itemize}
\item The top gate is an $\AND$ gate and the bottom-layer gates are $\OR$ gates.
%\item There are $2\ourd$ alternating layers of $\AND$ and $\OR$ gates, 
%  where the top and bottom-layer gates are $\AND$ and $\OR$ gates, respectively. 
%(So there are $\ourd$ layers of $\AND$ and $\OR$ gates each.)
\item The top $\AND$ gate has fanin $u$.
\item The bottom-layer $\OR$ gates all have fan-in $\red{w^{33/100}}$. 
%, except bottom-layer $\OR$ gates which have fan-in $\sqrt{w}$.  
\end{itemize}
%In other words, $\SkewedSipser_{u,\ourd}$ is the $\OR$ of $w$ disjoint copies of $\SkewedSipser'_{u,\ourd}$. So $\SkewedSipser_{u,\ourd}$ is over %$N(\ourd) := 
%$(uw)^\ourd  \sqrt{w}$ variables, and $\SkewedSipser'_{u,\ourd}$ is over %$N'(\ourd) := 
%$(uw)^{\ourd-1} \cdot u\sqrt{w}$ variables. 
\end{definition}

%\rnote{We may not even need the following notation (I'm not sure)...the new plan is that we %will only ever 
% take $\ell=1$ and $D=d$ -- this seems (perhaps surprisingly, at least to me) to give us %the optimal bounds.  I didn't get rid of any of the $\ell$'s below -- once we are sure of %this we will need to make a lot of little changes below because of this, there are notes %about some of them in what follows.}
For $\SkewedSipser_{u,\ourd}$ and each $\ell \in [\ourd+1]$, we write $\OR^{(\ell)}$ to denote an $\OR$ gate that is in the $\ell$-th level of $\OR$ gates away from the input variables and similarly write $\smash{\AND^{(\ell)}}$ to denote an $\AND$ gate that is in the 
  $\ell$-th level of $\AND$ gates away from the input variables. So the root of $\SkewedSipser_{u,\ourd}$ is the only $\smash{\OR^{(\ourd+1)}}$ gate;
% whereas the root of $\SkewedSipser'_{u,\ourd}$ is an $\AND^\ourd$ gate, and in both these formulas, 
  each $\smash{\AND^{(\ell)}}$ gate has $u$ many $\smash{\OR^{(\ell)}}$ gates as its inputs;
 each $\smash{\AND^{(1)}}$ gate of $\SkewedSipser_{u,\ourd}$ computes a
  disjoint copy of $\d2sipser_u$.

Next we introduce an addressing scheme for 
  gates and variables of $\SkewedSipser_{u,\ourd}$.\vspace{-0.2cm}

\paragraph{Addressing scheme.}
Viewing $\SkewedSipser_{u,\ourd}$ as a tree (with its leaves being variables
  and the rest being $\AND,\OR$ gates),
  we index its nodes (gates or variables) by addresses as follows.
The root (gate) is indexed by $\epsilon$, the empty string.
The $j$-th child of a node is indexed by the address of its parent concatenated with $j$. 
Thus,
%\rnote{We will still want this addressing scheme I think.} 
the variables of $\SkewedSipser_{u,\ourd}$ are indexed by addresses
\[ A(\ourd) := \Big\{ (b_0,a_1,b_1,\ldots,a_\ourd,b_\ourd) \colon a_i \in [u], b_0,\ldots,b_{\ourd-1} \in [w], b_\ourd \in [\red{w^{33/100}}] \Big\}. \] 
%Likewise, the variables of $\SkewedSipser'_{w,\ourd}$ are indexed by addresses
%\[ A'(\ourd) := \big\{ (a_1,b_1,\ldots,a_\ourd,b_\ourd) \colon a_i \in [u], %b_1,\ldots,b_{\ourd-1} \in [w], b_\ourd \in [\sqrt{w}] \big\}. \] 

\paragraph{Block and  {section} decompositions.}  
%\rnote{Maybe now we can get away with just talking about ``blocks'' instead of ``$\ell$-blocks'', i.e. specialize the discussion to $\ell=1.$}
%As each $\AND^\ell$ gate of $\SkewedSipser_{u,\ourd}$ is the output gate of a $\Sipser_{w, \ell}$ subfunction.  
\hspace{-0.2cm}We will  refer to the set of \orange{$u\red{w^{33/100}}$} addresses of variables below an $\AND^{(1)}$ gate as a \emph{block},
  and the set of \orange{$\red{w^{33/100}}$} addresses of variables below an $\OR^{(1)}$ gate as a \emph{section}.

It will be convenient for us to view the set of all variable addresses $A(\ourd)$ as
\begin{gather*} A(\ourd) = B(\ourd) \times A',  \text{~where} \\
 B(d) = \big\{ (b_0, a_1,b_1 \ldots, a_{\ourd-1}, b_{\ourd-1}) \colon  a_i \in [u], b_i \in [w] \big\}\ \ \ \text{and}\ \ \ A'=[u]\times [\red{w^{33/100}}].\end{gather*}
Here
 $B(d)$ can be viewed as the set of addresses of the $\AND^{(1)}$ gates of $\SkewedSipser_{u,\ourd}$,
 and $A'$ can be viewed as the set of  variable 
 addresses of $\d2sipser_u$ computed by 
 each such gate (following the same addressing scheme).
 %may be viewed as the set of addresses of all the $\AND^\ell$ gates of $\SkewedSipser_{u,\ourd}$.
 
More formally, for a fixed $\beta \in B(\ourd)$ %(specifying a particular $\AND^1$ gate of $\SkewedSipser_{u,\ourd}$), 
  we call the set of addresses 
 \[ A(\ourd,\beta) \eqdef \big\{ (\beta,\tau) \colon \tau \in A' \big\} \] 
 a \emph{block} of $A(\ourd)$; these are the addresses of variables below the $\AND^{(1)}$ gate specified by $\beta$. Thus, $A(\ourd)$ is the disjoint union of $w (uw)^{\ourd-1}$ many
  blocks, each of cardinality $|A' |=u\red{w^{33/100}}$. 
%\[ A(\ourd) = \bigsqcup_{\beta \in B(\ourd,\ell) } A(\ourd,\beta). \] 
%We call this the \emph{$\ell$-block decomposition} of $A(\ourd)$. 

For a fixed $\beta\in B(\ourd)$ and $a\in [u]$, %(with $\beta\circ a$ specifying a particular
  %$\OR^1$ gate of $\SkewedSipser_{u,\ourd}$), 
  we call the set of addresses
$$
\orange{A(\ourd,\beta, a)} \eqdef \big\{(\beta,a,b)\colon b\in [\red{w^{33/100}}]\big\}
$$
a \emph{section} of $A(\ourd)$; these are the addresses of variables below the
  $\OR^{(1)}$ gate specified by $\orange{(\beta, a)}$.
Each block $A(\ourd,\beta)$ is the disjoint union of $u$ many sections,
  each of cardinality $\red{w^{33/100}}$.
 
%\paragraph{Column decomposition of blocks.} \rnote{I t%hink we still want this decomposition; I'm not sure ``column'' is the best term, maybe now there is a better one?} Consider $A'(\ell)$, which we think of as an $\ell$-block $A(\ourd,\beta) \subset A(\ourd)$ for some $\beta \in B(d,\ell)$. It will be convenient for us to view $A'(\ell)$ as the disjoint union of $u^{\ell}$ many \emph{columns}: 
%\[ A'(\ell) = \bigsqcup_{a \in [u]^\ell} A'(\ell)_a \]
%where  
%\[ A'(\ell)_a = \big\{ (a_1,b_1,\ldots,a_\ell,b_\ell) \colon b_1,\ldots,b_{\ell-1} \in [w], %b_\ell \in [\sqrt{w}]\big\}.  
%\] 
%(As intuition for the term ``column'', recall that the set $A'(\ell)$ corresponds to the %set of all edges in the graph $G(\SkewedSipser_{u,\ell})$.  A column $A'(\ell)_a$ %corresponds precisely to a ``column'' of edges in the graph $G(\SkewedSipser_{u,\ell})$ %when the graph is laid out as depicted in Figure~\red{XXX}.)

To summarize, the set of addresses of variables $A(\ourd)$  
  can be decomposed into $w(uw)^{\ourd-1}$ many blocks $A(\ourd,\beta)$ (corresponding to the $\AND^{(1)}$ gates), $\beta\in B(d)$, and each such block can be further 
  decomposed into $u$ many sections $A(\ourd,\beta, a)$ (corresponding to its $u$ input $\OR^{(1)}$ gates), $a\in [u]$. 
  
Accordingly we also decompose $A'$, the set of variable addresses of $\d2sipser_u$,
  into sections $$A'(a)\eqdef \big\{(a,b)\colon b\in [\red{w^{33/100}}]\big\},\quad\text{for $a\in [u]$}.$$

The following fact is trivial given the definition of 
  $\d2sipser_{u}$. \orange{(Below and subsequently, we use ``$\varrho$'' to denote a restriction to the variables of $\d2sipser$ and ``$\rho$'' to denote a restriction to the variables of $\SkewedSipser$.)}

%The next proposition says that we can kill a $\SkewedSipser'_{u,\ell}$ function to the constant $0$ function by setting all the variables in a column of $A'(\ell)$ to $0$. Equivalently, we can kill the $\SkewedSipser'_{u,\ell}$ subfunction of $\SkewedSipser_{u,\ourd}$, say at address $\beta \in B(\ourd,\ell)$, by setting all the variables in a column of $A(\ourd,\beta)$ to $0$.) 

\begin{fact}
\label{prop:column-kill} 
For any $a\in [u]$ and restriction $\varrho \in \{0,1,\ast\}^{A'}$ that sets all    
  variables in the $a$-th section $A'(a)$ to $0$, i.e., 
  $\varrho_{\tau} = 0 \text{ for all $\tau \in A'(a)$},$ 
we have that $\d2sipser_u \uhr \varrho \equiv 0$.  
\end{fact} 

%The proof is an easy induction on $\ell$; alternately, one can observe that removing any column of edges in $G(\SkewedSipser'_{u,\ell})$ disconnects $s$ from $t$.

Now we define our random projection operator $\proj_\brho(\cdot)$.

\begin{definition}[Projection operators]
Given a restriction $\rho\in \{0,1,\ast\}^{A(d)}$, 
  the projection operator $\proj_\rho$ maps a function 
  $\smash{f \colon \{0,1\}^{A(d)}\rightarrow \{0,1\}}$
  to a function $\smash{\proj_\rho(f) \colon \{0,1\}^{B(d)}\rightarrow \{0,1\}}$, where
$$
\big(\proj_\rho(f) \big)(y)=f(x),\quad\text{where}\ x_{\beta,\tau} \eqdef\
\begin{cases} y_{\beta} &\text{if 
  $\rho_{\beta,\tau}=\ast$} \\ 
\rho_{\beta,\tau} & \text{if $\rho_{\beta,\tau}\in \{0,1\}$.}
  \end{cases}
$$
For convenience, we sometimes write $\proj(f \uhr \rho)$ instead of $\proj_\rho(f)$.
\end{definition}
\begin{remark}
The following interpretation of the projection operator will come in handy.
Given~a restriction $\rho\in \{0,1,\ast\}^{A(d)}$, if $f$ is computed by a circuit $C$, then
  $\proj_\rho(f)$ is computed by a circuit $C'$ obtained 
  from $C$ by replacing every occurrence of $x_{\beta,\tau}$ by $y_\beta$
  if $\rho_{\beta,\tau}=\ast$, or by $\rho_{\beta,\tau}$ if $\rho_{\beta,\tau}\in \{0,1\}$.
\end{remark}

The crux of our \emph{random} projection 
  operator $\proj_\brho(\cdot)$ is then a distribution $\calD^{(d)}_{u}$ over 
  restrictions $\smash{\{0,1,\ast\}^{A(d)}}$ to the variables 
  $\{x_{\beta,\tau}:(\beta,\tau)\in A(d)\}$, 
  from which $\brho$ is drawn.
To this end, we consider the block decomposition $B(\ourd ) \times A'$ of 
  $A(\ourd)$, and $\smash{\brho\leftarrow \calD^{(d)}_u}$ is obtained by
  drawing independently, for each block $\beta\in B(d)$, a restriction $\brho_\beta$ from a distribution
  $\calD_u$ over $\{0,1,\ast\}^{A'}$ to be defined below.
  
%  , and our overall random projection of its variables $$\big\{ x_{\beta,\tau} \colon \beta \in B(\ourd), \tau \in A' \big\}$$
%   will be independent across different blocks. It will be $\proj_\brho (\cdot)$ where 
%\begin{itemize}
%\item 
%A restriction $\brho \in \{0,1,\ast\}^{B(\ourd )\times A'(\ell)}$ is first drawn
%  from $\calD_{u}^{(\ourd)}$, where $\brho_\beta$ for each
%  block $\beta\in B(\ourd)$ is drawn independently from a distribution $\calD_u$ 
%  over $\{0,1,\ast\}^{A'}$ to be defined later.
%\ignore{ to the distribution $\calD_{u }^{(\ourd)}$}
%\[ \brho_\beta \leftarrow \calD_{u },\quad \text{independently for each $\beta \in %B(\ourd )$}. \] 
%(We define the distribution $\calD_{u }$ over $\{0,1,\ast\}^{A' }$ below).
%Given $\brho$, each $x_{\beta,\tau}$ is replaced by $\rho_{\beta,\tau}$ if $\rho_{\beta,%\tau}\in \{0,1\}$,
%  and remains to be $x_{\beta,\tau}$ if $\rho_{\beta,\tau}=*$.
%\item Given $\brho$, the projection operator $\proj_\brho$ does a ``promotion'': 
%\[ x_{\beta,\tau} \mapsto y_\beta, \quad \text{for all $\beta\in B(\ourd)$ and $\tau \in %A'$ such that $\rho_{\beta,\tau}=*$.}\] 
%Intuitively, all variables $x_{\beta,\tau}$ (that remain after the 
%  random restriction $\brho$) below a particular $\AND^\ell$ gate $\beta$
%are identified with a new variable $y_\beta$ corresponding to that $\AND^\ell$ gate.
%\end{itemize} 

%To finish the definition of our random projection, we define $\calD_u$.

\begin{definition}[Distributions $\calD_{u }$ and $\calD_u^{(d)}$]
\label{def:D_ell}
The distribution $\calD_{u} = \calD_{u}(q)$ over $\{0,1,\ast\}^{A'}$ is para\-meterized by a probability $q\in (0,1)$. A draw of a restriction $\boldsymbol{\varrho}$ from $\calD_{u}$ is generated as follows:
\begin{flushleft}\begin{itemize}
\item With probability $q$, output $\boldsymbol{\varrho} = \{\ast\}^{A'}$ (i.e. the restriction fixes no variables).
\item Otherwise (with probability $1-q$), we draw $\ba \leftarrow [u]$ (a random section) and $\bz \leftarrow \{0,1\}$ (a random bit) independently and uniformly at random, and output $\boldsymbol{\varrho}$ where 
for each $\tau \in A'$,
\[ 
\boldsymbol{\varrho}_\tau =  
\begin{cases} 
\bz & \text{if $\tau \in A'(\ba)$} \\
1-\bz & \text{otherwise.} 
\end{cases}
\] 
Note that in this case $\boldsymbol{\varrho}$ is distributed uniformly among $2 u$ many binary strings in $\zo^{\orange{A'}}$. These strings are ``section-monochromatic'', with $u-1$ of the sections taking on entirely the same value $\orange{1 - \bz}$ and the one remaining section $\ba$  taking entirely the other ``rare'' value $\orange{\bz}$. 
\end{itemize} \end{flushleft}
%\end{definition}

As described above, a draw of $\brho \in \{0,1,\ast\}^{B(\ourd )\times A' }$ from $\calD_{u}^{(\ourd)}=\calD_{u }^{(\ourd)}(q)$
is obtained by independently drawing $\brho_\beta \leftarrow \calD_{u}=\calD_u(q)$ for each   block $\beta \in B(\ourd).$
\end{definition} 

The following observation about $\supp(\calD^{(d)}_u)$ will be useful for us: 

\begin{remark}
\label{rem:obvious} 
A restriction $\rho \in \{0,1,\ast\}^{B(d) \times A'}$ is in the support of $\calD^{(d)}_u$ iff for every block $\beta \in B(d)$, $\rho_\beta$ is either $\{ \ast\}^{A'}$, or there exists exactly one section $a \in [u]$ such that $\rho_{\beta,\tau} = 0$ if $\tau \in A'(a)$ and $1$ otherwise, or there exists exactly one section $a \in [u]$ such that $\rho_{\beta,\tau} = 1$ if $\tau \in A'(a)$ and $0$ otherwise. 

Therefore, if $T$ is a term of width at most $u -1$ such that for all blocks $\beta \in B(d)$, the variables from block $\beta$ that occur in $T$ all occur with the same sign, then $T$ can be satisfied by a restriction in the support of $\smash{\calD^{(d)}_u}$ (i.e.,~$T\uhr \rho \equiv 1$ for some $\smash{\rho \in \supp(\calD^{(d)}_u)}$). (Note that this crucially uses the fact that $T$ has width at most $u-1$, and in particular  does not contain variables from all $u$ sections of any block $\beta$. Also note that
  the inverse of this is not true, e.g.,~consider $T=x_{\beta,\tau}\land \neg\, x_{\beta,\tau'}$ with
  $\tau$ and $\tau'$ from two different sections.)
\end{remark}

\section{Projection Switching Lemma} \label{sec:projection_switching_lemma}

Our goal now is to prove the following projection switching lemma for (very) small width DNFs:  
%\rnote{We now think we'll only be using this projection switching lemma in the $\ell=1$ case.  I think this actually won't lead to that much simplification in the proof; if it is really the case that we only use the $\ell=1$ case, though, then probably it doesn't make sense to state/prove it in the more general version.  I didn't change it (it is still the general-$\ell$ version below), let's wait until the new stuff is a bit more stable.}

\begin{theorem}[Projection Switching Lemma] 
\label{theorem:PSL} %\rnote{If we keep the statement in this form with a general $\ell$, should we quantify over $\ell$ in the theorem statement?  Probably we won't keep it in this form though.} 
For $2 \leq u \leq w$, let $F$ be an $r$-\emph{DNF} over the variables $\{x_{\beta,\tau}\}$, ${(\beta,\tau)\in A(\ourd)}$, where $r \le u -1$.  Then for all $s\ge 1$ and $q\in (0,1)$, we have
\[ \Prx_{\brho\leftarrow\calD_u^{(\ourd)}\orange{(q)}} \Big[\hspace{0.06cm}\text{$\proj_\brho(F)$ has decision tree depth $\ge s$}\hspace{0.03cm}\Big] \le \left(\frac{8qru}{1-q}\right)^{\! s}.\]  
\end{theorem}

Notice that while $F$ is an $r$-DNF over formal variables $\{ x_{\beta,\tau} \colon 
  (\beta,\tau) \in A(d)\}$, we will bound the decision tree depth of $\proj_\brho(F)$, a function over the new formal variables $\{ y_\beta \colon \beta \in B(\ourd) \}$.  
\begin{remark} \label{rem:1}
Projections will play a key role in the proof. Consider a term of the form
$T = x_{\beta,\tau} \wedge \neg\, x_{\beta,\tau'}$
for some $\tau \neq \tau'$,
and suppose our $\rho$ from $\calD^{(\ourd)}_{u}$ is such that $\rho_{\beta,\tau} = \rho_{\beta,\tau'} = \ast$. In this case we have  
$T\uhr \rho = x_{\beta,\tau} \wedge \neg\, x_{\beta,\tau'},$
i.e., the term survives the restriction $\rho$, but  
$\proj_\rho(T) %= \proj(T\uhr \rho) 
= y_\beta \wedge \neg\,y_\beta \equiv 0,$
i.e.,~the term is killed by
  $\proj_\rho$. Our proof will crucially leverage simplifications of this sort.  
\end{remark}

\begin{remark}\label{rem:2}
The parameters of Theorem~\ref{theorem:PSL} are quite delicate in  the sense that the statement 
  fails  to hold for DNFs of width $u$. 
To see this, consider $\SkewedSipser_{u,\orange{d}}$ with $d=1$, a depth-$3$ formula that can also be written as a $u$-DNF. Then by Corollary~\ref{cor:preserve-target} (to be introduced in Section \ref{sec:proof_depth_hierarchy}), we have that for $\smash{\brho \leftarrow \calD^{(1)}_{u}(q)}$ with $\smash{q = w^{-669/1000}}$,  
  the function $\smash{\proj_\brho(\SkewedSipser_{u,1})}$ 
  contains a $\smash{\red{w^{33/100}}}$-way $\OR$ 
%(the function
%$\SkewedSipser_{2,0}$) 
as a subfunction --- and hence has decision tree depth at least  $\red{w^{33/100}}$ --- with probability $1-o(1)$.  
%  a $2^{\ourd}$-width DNF.\rnote{If we specialize to the case $\ell=1$, do we need to revisit this example?}  By Corollary~\ref{cor:preserve-target}, we have that for $\brho \leftarrow \calD^{(\ourd)}_{u}(q)$, where $q = w^{-1/3}$, the function $\proj^{\ourd}_\brho(\SkewedSipser_{2,\ourd})$ contains a $\sqrt{w}$-way $\OR$ (the function
%$\SkewedSipser_{2,0}$) as a subfunction --- and hence has decision tree depth at least  $\sqrt{w}$ --- with probability $1-o(1)$.  
So while the statement of Theorem \ref{theorem:PSL} holds for $(u-1)$-DNFs, it
  does not hold for $u$-DNFs when $u=o(w^{\red{669/2000}})$ and $w \to \infty$.
\end{remark}
\begin{remark} \label{rem:CNF_case}
We observe that the conclusion of Theorem \ref{theorem:PSL} still holds if the condition ``$F$ is an $r$-CNF'' replaces ``$F$ is
an $r$-DNF.''  This can be shown either by a straightforward adaptation of our proof, or via a reduction to the DNF case using duality, the invariance of our distribution of random projections under the operation of flipping each bit, and the fact that decision tree depth does not change when input variables and output value are negated.
\end{remark}

\subsection{Canonical decision tree} 

Given an $r$-DNF $F$ over variables $\{x_{\beta,\tau}:(\beta,\tau)\in A(d)\}$ 
  and a restriction $\rho \in \{0,1,\ast\}^{A(\ourd)}$, $\proj_\rho(F)$ is a function over the new variables $\{y_\beta:\beta\in B(d)\}$. \violet{We assume a fixed but arbitrary ordering~on the terms in $F$, and the variables within terms.} 
The canonical decision tree $\mathsf{CanonicalDT}(F, \rho)$ 
  that computes $\proj_\rho(F)$ is defined inductively as follows.\bigskip 

%Below we assume a fixed ordering on the terms of $F$ and a fixed ordering on 
%  the variables within each term of $F$, and use $i$ only to index the terms processed by the procedure (as we shall see, this indexing will be convenient in the proof of Theorem \ref{theorem:PSL}). \bigskip 

\noindent {\sf CanonicalDT}$\orange{(F,\rho)}$\hspace{0.06cm}: 
%\begin{flushleft}
\begin{enumerate}
\item[0.] If $\orange{\proj_\rho(F)} \equiv 0$ or $1$, output $0$ or $1$, respectively. 
%\orange{(Otherwise, at least one term $T$ of $F$ satisfies $T\not\equiv 0$
%  and no term $T$ of $F$ satisfies $T\equiv 1$.)}

\item \violet{Otherwise, let $T$ be the first term in $F$ such that $T \uhr \rho$ is non-constant and $T \uhr \rho\rho' \equiv 1$~for some $\smash{\rho'\in \supp(\calD_u^{(\ourd)})}$.  We observe that such a term must exist, or the procedure would have halted at step 0 above and not reached the current step 1. \vspace{0.04cm} 

To see this, first note that certainly there must exist a term $T'$ such that $T' \uhr \rho$ is non-constant since otherwise $F \uhr \rho$ is constant (and likewise $\proj_\rho(F)$). We furthermore claim that among these terms $T'$, there must exist one such that $T' \uhr \rho$ is satisfiable by some $\smash{\rho'\in \supp(\calD_u^{(d)})}$,~i.e. $T'\uhr\rho\rho'\equiv 1$. To prove this, suppose that each of these terms $T'$ satisfies that $T' \uhr \rho$~is~non-constant and there exists no restriction $\smash{\rho' \in \supp(\calD^{(d)}_u)}$ such that $T' \uhr \rho\rho' \equiv 1$.  By Remark~\ref{rem:obvious} (and our assumption that $r \le u-1$), $T' \uhr \rho$ must contain two literals from the same block occurring with opposite signs, i.e.,~$x_{\beta,\tau}$ and $\neg\, x_{\beta,\tau'},$ for some $\beta \in B(d)$. In this case, we have that $\proj_\rho(T')$ contains both $y_\beta$ and $\neg\,y_\beta$ and hence $\proj_\rho(T') \equiv 0$. But if each such term $T'$ has $\proj_\rho(T')\equiv 0$, then $\proj_\rho(F)\equiv 0$ and the procedure would have halted at step (0).}

\item Define 
\[ \eta = \big\{ \beta\in B(d) \colon x_{\beta,\tau} \text{~or~}\neg\,x_{\beta,\tau}\text{ occurs in $T \uhr \rho$
for some $\tau$}\big\} \] 
Our canonical decision tree will then query variables $y_\beta$, $\beta\in \eta$ exhaustively, i.e.,~we grow a complete binary tree of depth $|\eta|$; \red{we will refer to $T$ as \emph{the
  term} of this tree}. 
%\gray{For the example above, our canonical decision tree will query 
%  $y_\beta,y_{\beta'},$ and $y_{\beta''}$.\ignore{ and $y_{\beta'''}$.}}

\item For every assignment $\pi \in \zo^{\eta}$ to variables $y_\beta$, $\beta\in \eta$ (equivalently, every path through the complete binary tree of depth $\smash{|\eta|}$), we recurse on $\mathsf{CanonicalDT}(F, \rho(\eta \mapsto \pi))$, where we use $(\eta \mapsto \pi) \in \{0,1,\ast\}^{A(d)}$ to denote the following restriction:
\begin{equation}\label{def:res} (\eta \mapsto \pi)_{\beta,\tau} = 
\begin{cases} 
%\red{\pi}(y_\beta) 
\red{\pi_\beta} & \text{if $\beta \in \eta$}  \\
\ast & \text{otherwise,} 
\end{cases}\quad\text{for all $\beta\in B(d)$ and $\tau\in A'$}.
\end{equation} 
\end{enumerate}
%\end{flushleft}

\begin{proposition} 
\label{prop:CDT-is-DT}
For every $\rho \in \{0,1,\ast\}^{A(\ourd)}$, we have that 
$\mathsf{CanonicalDT}(F,\rho)$ computes $\proj_\rho(F)$. 
\end{proposition} 

While $\mathsf{CanonicalDT}$ is well defined for all $\rho$, we shall mostly be interested in $\rho \in \supp(\calD^{(\ourd)}_{u})$.

\subsection{Proof of Theorem \ref{theorem:PSL}}

Let 
\[ \calB \eqdef \big\{ \rho \in \supp(\calD^{(\ourd)}_{u}) \colon \text{decision tree depth of\ }\mathsf{CanonicalDT}(F,\rho) \ge s \big\} \]
 be the set of bad restrictions,  To prove Theorem~\ref{theorem:PSL}, it suffices to bound
\ignore{\begin{equation}\label{probtobound}}
$ 
 {\Pr_{\brho\leftarrow\calD^{(\ourd)}_{u}\orange{(q)}}[\hspace{0.03cm}\brho\in \calB\hspace{0.03cm}]},
$
\ignore{\end{equation}}
 the total weight of $\calB$ under $\calD^{(\ourd)}_{u}\orange{(q)}$.  
 Following Razborov's strategy (see \cite{beame1993switching} for more details), we will construct a map $$\theta
 \colon \calB\rightarrow \{0,1,\ast\}^{A(\ourd)} \times \zo^s \times \zo^{s (\log r + 1)},$$
%%\begin{equation} \label{eq:def_encoding}
%%\theta(\rho) = \big(\rho \sigma^{(1)} \cdots \sigma^{(s')},\hspace{0.08cm} \encode(\pi),%%\hspace{0.08cm} \text{$t$ bits of ``auxiliary information''}\big),
%%\end{equation} 
\ignore{
\rnote{above, write ``$t$ bits of ``auxiliary information'''' or ``$\encode(\eta)$''? \violet{{\bf Li-Yang:} Yeah all of this was written super informally and sloppily (e.g.~$s'$ is undefined, I didn't address the issue of truncating $\pi$ to have length exactly $s$, etc.). You are right, unlike in RST15 where our auxiliary information was significantly more complicated, in the above it is just $\encode(\eta)$. So if we want to be vague both $\encode(\pi)$ and $\encode(\eta)$ should be considered ``auxiliary information'', and if we want to be precise we should just write $\encode(\pi)$ and $\encode(\eta)$. What I have written above is neither here nor there. I'll rewrite this whole section at some point...}}
}
%for some $s' \le s$ and some ``small'' $t \le s(1+\log r)$, 
with the following two key properties:
\begin{enumerate}
\item {\bf (injection)} $\theta(\rho) \ne \theta(\rho')$ for any two distinct 
  restrictions $\rho,\rho'\in \calB$;
\item {\bf (weight increase)} Let $\theta_1(\rho)\in \{0,1,\ast\}^{A(d)}$ denote the 
  first component of $\theta(\rho)$. Then
\begin{equation}\label{juju}
 \frac{\Pr [\hspace{0.03cm}\brho = \theta_1(\rho)\hspace{0.03cm}]}{\Pr\hspace{0.03cm}[\brho  = \rho\hspace{0.03cm}]} \ge \Gamma, \quad
\text{for all $\rho \in \calB$},
\end{equation} %where $\sigma = \sigma^{(1)} \cdots \sigma^{(s')}$ and 
where $\Gamma = \orange{((1-q)/2qu)^s}$ is ``large''. 
\end{enumerate} 

Assuming such a map $\theta$ exists (below we describe  its construction
  and prove the two properties stated above), Theorem \ref{theorem:PSL} follows
  from a simple combinatorial argument.
\begin{proof}[Proof of Theorem \ref{theorem:PSL}]
Fix a pair $O \in \zo^{s} \times \zo^{s(1+\log r)}$ and let %of $(\encode(\pi), \encode(\eta))$ and consider 
\[ \calB_O = \big\{ \rho \in \calB \colon \big(\theta_2(\rho),\theta_3(\rho)\big) = O\big\} \sse \calB,  \]
where we use $\theta_2(\rho)$ and $\theta_3(\rho)$ to denote the
  second and third components of $\theta(\rho)$, respectively.

Then we have that 
\begin{align*} 
\Prx
%_{\calD^{(\ourd)}_{u}}
[\hspace{0.03cm}\brho\in \calB_O\hspace{0.03cm}] = \sum_{\rho\in\calB_O}  
  \Pr[\hspace{0.03cm}\brho = \rho\hspace{0.03cm}]  
  \le (1/\Gamma)\cdot \sum_{\rho\in\calB_O} \Pr[
  \hspace{0.03cm}\brho = \theta_1(\rho)\hspace{0.03cm}] 
  \le 1/\Gamma.
\end{align*}
Here the first inequality uses (\ref{juju}) and
  the second inequality uses the property of $\theta$ being an injection:
 % and the second inequality is by %Lemma~\ref{lemma:injection}: 
  %since the map $\theta$ is an injection, 
  we have that $\theta_1(\rho) \ne \theta_1(\rho')$ for any two distinct $\rho,\rho'\in \calB_O$ (recall that $\theta_2(\rho) = \theta_2(\rho')$ and $\theta_3(\rho) = \theta_3(\rho')$), and therefore $\sum_{\rho\in\calB_O} \Pr[\hspace{0.03cm}\brho = \theta(\rho)_1\hspace{0.03cm}]\le 1$.
Summing up over all possible $O$'s, we have
% \in \zo^s \times \zo^{s(1+\log r)}$ gives us 
\[ \Pr[\hspace{0.03cm}\brho\in\calB\hspace{0.03cm}] = 
\sum_{O} \Pr[\hspace{0.03cm}\brho\in \calB_O\hspace{0.03cm}] \le 2^s \cdot (2r)^s \cdot \big(2qu/(1-q)\big)^s = \big(8qru/(1-q)\big)^s,  \] 
and this concludes the proof of Theorem~\ref{theorem:PSL}.
\end{proof}
%A simple combinatorial argument \orange{presented later} then allows us to bound 
%  (\ref{probtobound}) from above by
%%$\Pr[\brho \in \calB]$ %(and hence the failure probability of our \orange{PSL}) by 
%\[ \frac{\text{``encoding cost of $\pi$ and auxiliary information''}}{\text{weight %increase}} = \frac{2^{s} \cdot 2^t}{\Gamma} = \big(2 \cdot (2r) \cdot 2qu^{\ell}\big)^s.  %\]

The rest of the section is organized as follows.
We construct the map $\theta$ in Section \ref{sec:encoding}.
Then we show that it is an injection in Section \ref{sec:decodability},
  by showing that one can decode $\rho$ from $\theta(\rho)$ uniquely~for 
  any $\rho\in \calB$.
Finally we prove the weight increase, i.e., (\ref{juju}) in Section \ref{sec:weight}.

\subsection{Encoding}\label{sec:encoding}

Let $\rho\in \calB$ be a bad restriction.
Let $\pi^*$ be the lexicographically first path of length at least $s$ in the decision tree  $\mathsf{CanonicalDT}(F,\rho)$ (witnessing the badness of $\rho$),
  and $\pi$ be its truncation at length $s$.
Then $\theta_2(\rho)$ is defined to be $\textrm{binary}(\pi)\in \{0,1\}^s$,
  the binary representation of $\pi$, i.e., $\pi_i\in \{0,1\}$ is the evaluation
  of the $i$th $y$-variable along $\pi$.
  
Recall that $\mathsf{CanonicalDT}(F,\rho)$ is composed of a collection of
  complete binary trees, one for~each recursive call of $\mathsf{CanonicalDT}$.
Let $R_1,\ldots,R_{s'}$ for some $1\le s'\le s$ 
  denote the sequence of~complete binary trees that $\pi$
  visits, with $R_1$ sharing the same root as $\mathsf{CanonicalDT}(F,\rho)$
  and $\pi$ ending in $R_{s'}$. (Here $s'\ge 1$ because $s\ge 1$.)
We also use $T_i$ to denote the term of tree $R_i$, for each $i\in [s']$. 

For each $i\in [s'-1]$, we let 
$$
\eta_i=\big\{\beta\in B(d): y_\beta\ \text{is queried in tree $R_i$}\big\},
$$
and for the special case of $i=s'$, we let
\begin{equation}
\eta_{s'}=\big\{\beta\in B(d): y_\beta\ \text{is queried in tree $R_{s'}$
  before the end of $\pi$}\big\}.\label{eq:special-case}
 \end{equation}  
For each $i\in [s']$, $\pi$ induces a binary string $\pi^{(i)}\in \{0,1\}^{\eta_i}$,
  where $\pi^{(i)}_\beta$ for each $\beta\in \eta_i$ is set to be the evaluation of
  $y_\beta$ along $\pi$ (in tree $R_i$). Note that $T_i$ is the $i$-th term processed by $\mathsf{CanonicalDT}(F,\rho)$ along the bad path $\pi$ and equivalently, $T_i$ is the first term processed by $$\mathsf{CanonicalDT}\big(F,\rho(\eta_1 \mapsto \pi^{(1)})\cdots (\eta_{i-1}\mapsto \pi^{(i-1)})\big),$$ where $(\eta_j\mapsto \pi^{(j)})$ is a restriction defined as in (\ref{def:res}). So $T_i$ is the first term in $F$ such that $$T_i \uhr \rho(\eta_1 \mapsto \pi^{(1)})\cdots (\eta_{i-1}\mapsto \pi^{(i-1)})$$ is non-constant and $$T_i \uhr \rho(\eta_1 \mapsto \pi^{(1)})\cdots (\eta_{i-1}\mapsto \pi^{(i-1)}) \rho' \equiv 1,\quad \text{for some $\rho' \in \supp(\calD^{(d)}_u)$.}$$

At a high level, $\theta_1(\rho)$ and $\theta_3(\rho)$ 
  are defined as follows. 
The third component $$\theta_3(\rho)=\encode(\eta_1)\circ \cdots\circ \encode(\eta_{s'})\in 
  \{0,1\}^{s(1+\log r)}$$
  is the concatenation of $s'$ binary strings,
  where each $\encode(\eta_i)$ is a \emph{concise} \ignore{\rnote{Removed the word ``binary'' here -- the phrase ``binary representation'' makes me think of a number being represented, but here it is a set that we are representing as a bitstring in a concise way.}}representation
  of $\eta_i$. In particular, we are able to recover $\eta_i$ given both $\encode(\eta_i)$ and $T_i$.
We describe the encoding of $\eta_i$ in Section \ref{sec:eta}.
For the first component we have
$$\theta_1(\rho)=\rho \sigma^{(1)}\cdots\sigma^{(s')}\in \{0,1,\ast\}^{A(d)},$$ where
  each $\sigma^{(i)}\in \{0,1,\ast\}^{A(d)}$  
  is a restriction and $\rho \sigma^{(1)}\cdots\sigma^{(s')}$
 %\rnote{This was ``$\rho$''} 
  is their composition
  %(see the definition at the end of Section ?)
  %of $ s'+1 $ restrictions 
(note that each of these $s'+1$ restrictions,
  like the overall composition, belongs to $\{0,1,\ast\}^{A(d)}$).\ignore{(instead of their concatenation).}
We define the~$\sigma^{(i)}$'s in Section \ref{sec:general}. %, and represents a single restriction obtained by composition.) 
%Below we define $\smash{\sigma^{(i)}}$ and $\encode(\eta_i)$ for each $i$,
%  starting with the simpler base case of $i=1$.  

\subsubsection{Encoding $\eta_i$}\label{sec:eta}

%We have
%\[ \eta_1 =  \big\{ \beta\in B(d) \colon x_{\beta,\tau} \text{~or~}\neg\,x_{\beta,\tau}\text{ occurs in $T \uhr \rho$
%for some $\tau\in A'$}\big\}.\] % \equiv \{ \text{variables in $\proj_\rho(T_1)$}\}. \]
Fix an $i\in [s']$. 
Let $\eta_i=\{\beta_1,\ldots,\beta_t\}$ for some $t\ge 1$, with $\beta_j$'s ordered
  lexicographically.
%Sine $T_1$ has at most $r$ variables, we can use $\log r$ bits to encode the index
%  of the first $x_{\beta,\cdot}$ or $\neg\,x_{\beta,\cdot}$ variable that occurs
%  in $T_1$\footnote{\red{Xi: I think here we don't need $T_1\uhr \rho$ but $T_1$ is good enough.}}
%  (such a variable must exist due to Observation  
It follows from the definition of $\eta_i$ that 
  every $\beta_j$ appears in $T_i$, meaning that 
  either $x_{\beta,\tau}$ or $\neg\,x_{\beta,\tau}$ appears in $T_i$ for some $\tau\in A'$.%\footnote{\red{Xi: I think we only need to use $T_i$ instead of $T_i\uhr \rho \cdots$. This is
%  why I feel it is easier to separate out the definition of $\encode(\eta_i)$.}}

Instead~of encoding each $\beta_j$ directly using its binary representation,
  we use $\log r$ bits to encode the index of the first $x_{\beta,\cdot}$ or $\neg\,x_{\beta,\cdot}$
  variable that occurs in $T_i$. Here $\log r$ bits suffice because $T_i$ has at most $r$ variables. Also recall
  that we fixed an ordering on the variables of each term, so indices of variables
    in $T_i$ are well defined.
We let $\mathrm{location}(\beta_j)$ denote the $\log r$ bits for $\beta_j$.
We also append~it with one additional bit to indicate whether $\beta_j$ is the last
  element in $\eta_i$.
More formally, we write
$$
\encode(\eta_i)=\mathrm{location}(\beta_1)\circ 0 \circ \mathrm{location}(\beta_2)\circ 0
\circ \cdots\circ \mathrm{location}(\beta_t)\circ 1 \in \{0,1\}^{|\eta_{\orange{i}}|(1+\log r)}.
$$

We summarize properties of $\theta_3(\rho)$ below:

\begin{proposition}\label{simpleprop}
Given $\orange{\theta_3}(\rho)$, one can recover uniquely $s'$ and 
  $\encode(\eta_1),\ldots,\encode(\eta_{s'})$.
Furthermore, given $\encode(\eta_i)$ and $T_i$ for some $i\in [s']$, one can recover
  uniquely $\eta_i$.  
\end{proposition}

\ignore{
% START IGNORED BASE CASE

\subsubsection{$\sigma^{(1)}$: the base case}\label{sec:base}

\rnote{I know Li-Yang had thought about getting rid of Section 4.3.2, the base case, altogether and just doing the general case.  Should we do this?  I think having it there makes it just a little easier to read, but there
is a lot of redundancy.  If we get rid of it we should slightly modify Section 4.3.2 to incorporate the explanations from 4.3.1.}

%\begin{itemize} 

By the construction of $\mathsf{CanonicalDT}$, 
  $T_1$ is the first term in $F$ such that %$\proj(T_1 \uhr \rho) \not\equiv 0$ 
  $\proj_\rho(T_1)\not\equiv 0$. 
Thus,
\[ \eta_1 =  \big\{ \beta\in B(d) \colon x_{\beta,\tau} \text{~or~}\neg\,x_{\beta,\tau}\text{ occurs in $T \uhr \rho$
for some $\tau\in A'$}\big\}.\]
Also for each $\beta\in \eta_1$, variables in block $A(d,\beta)$
  that occur in $T_1\uhr \rho$ must occur with the same sign.

%(i.e.~the first term ``processed'' by $\mathsf{CanonicalDT}(F,\rho) \equiv \mathsf{CanonicalDT}(F,\rho,\orange{1})$; \orange{this term exists since  $\rho \in \mathcal{B}$ and $s \geq 1$}), and 
%For each $y_\beta \in \eta_1$, we use $\log(T_1) \le \log r$ bits to  encode the location of an arbitrary (say the first) $x_{\beta,\cdot}$ variable that occurs in $T_1 \uhr \rho$, along with one additional bit to denote whether $y_\beta$ is the last variable in $\eta_1$. We write $\encode(\eta_1) \in \zo^{|\eta_1| (1 + \log r)}$ to denote the overall string. 
The restriction $\smash{\sigma^{(1)} \in \{0,1,\ast\}^{\orange{B(d) \times A'}}}$ is constructed as follows. 
We set $\smash{\sigma^{(1)}_{\beta,\tau}=\ast}$, for all \orange{pairs $(\beta,\tau)$ with $\beta\notin \eta_1$ and $\tau\in A'$}.
For each  $\beta \in \eta_1$, \orange{let $a_\beta \in [u]$ be the smallest index such that 
  %no variable from the $(a_\beta)$th section of block $\beta$ or its negation appears in $T_1\uhr \rho$: 
  for every $\tau=(a_\beta,b)$ with $b \in [\red{w^{33/100}}]$, neither $x_{\beta,\tau}$
  nor $\neg\,x_{\beta,\tau}$ appears in $T_1 \uhr \rho.$}  ({Here is where we use our assumption~on the width of $F$:} Because $\mathrm{width}(T_1\uhr \rho) \le \mathrm{width}(T_1) \le r \le u - 1$ and there are $u$ sections in block $A(\ourd,\beta)$, there must exist such a section  ``untouched'' by $T_1 \uhr \rho$.)  %Recall that at least one variable from the block $A(\ourd,\beta)$   occurs in $T_1 \uhr \rho$, and all of them occur with
  % the same sign . 
If variables from the block $A(\ourd,\beta)$ all occur positively in $T_1 \uhr \rho$, we set 
\[ \sigma^{(1)}_{\beta,\tau} =  
\begin{cases}
1 & \text{if \orange{$(\beta,\tau) \notin A(\ourd,\beta,a_\beta)$}} \\
0 & \text{if \orange{$(\beta,\tau) \in A(\ourd,\beta, a_\beta)$},}
\end{cases} 
\] 
and if they all occur negatively in $T_1 \uhr \rho$ (note that one of these two cases must hold), we set 
\[ \sigma^{(1)}_{\beta,\tau} = 
\begin{cases}
0 & \text{if \orange{$(\beta,\tau) \notin A(\ourd,\beta,a_\beta)$}} \\
1 & \text{if \orange{$(\beta,\tau) \in A(\ourd,\beta, a_\beta)$}.}
\end{cases} 
\]  
%We do the above for all $\beta$ such that $y_\beta \in \eta_1$, and for the remaining $\beta'$ such that $y_{\beta'} \notin \eta_1$, we  
%set $\sigma^{(1)}_{\beta',\tau} = \ast$ for all \orange{$\tau \in A'$}. 
This finishes the definition of $\sigma^{(1)}$. We mention two key properties of $\sigma^{(1)}$: 
\begin{enumerate} 
\item $T_1 \uhr \rho \sigma^{(1)} \equiv 1$. (This will be useful for decoding.) 
\item For all $\beta \in \eta_1$, we have that $\rho_{\beta,\tau}=\ast$
  for all $\tau\in A'$ whereas \[ \sigma^{(1)}_\beta \in \zo^{A(\ourd,\beta)}, \]
\orange{where $\sigma^{(1)}_\beta$ (and $\rho_\beta$ below) 
  denotes the subrestriction of $\sigma^{(1)}$ (or $\rho$, respectively) 
    over $A(\ourd,\beta)$.} By the definition of $\calD_u\orange{(q)}$, we have  
\[ \Prx_{\boldsymbol{\varrho}\leftarrow\calD_{u}\orange{(q)}}[\hspace{0.03cm}\boldsymbol{\varrho} = \rho_\beta\hspace{0.03cm}] = q \quad \text{whereas} \quad \Prx_{\boldsymbol{\varrho}\leftarrow\calD_{u}\orange{(q)}}[\hspace{0.03cm}\boldsymbol{\varrho} = \sigma^{(1)}_\beta\hspace{0.03cm}] = \frac{1-q}{2u}.\] (This will be useful for the weight increase calculations.)
\ignore{\rnote{Two questions:  first, should it be
$\weight_{\calD_{u,\ell}}(\sigma^{(1)}_\beta) = \frac{1-q}{2^{\ell+1}}$?  \violet{{\bf Li-Yang:} Yup, fixed.} Second, did we define somewhere the notation that
$\weight_{\calD_{u,\ell}}(\kappa)$ denotes $\Pr_{\brho \leftarrow \calD_{u,\ell}}[\brho = \kappa]$?  If not we should say it somewhere. \violet{{\bf Li-Yang:} Yeah it wasn't defined. I just got rid of $\weight_{\calD_{u,\ell}}(\cdot)$ since we do not use it elsewhere and it is not much more succinct.}
}}

\end{enumerate} 
%\end{itemize} 

%\orange{Having defined $\sigma^{(1)}$, $\encode(\pi^{(1)})$, and $\encode(\eta_1)$, we proceed with the general case. The strings $\encode(\pi)$ and $\encode(\eta)$ will be simply the concatenation of the respective strings $\encode(\pi^{(i)})$ and $\encode(\eta_i)$, where $i \in [s']$. Finally, the ``auxiliary information'' alluded to in (\ref{eq:def_encoding}) corresponds to $\encode(\eta)$. }

% END IGNORED BASE CASE
}

\subsubsection{The $\sigma^{(i)}$ restriction}\label{sec:general}

%Let $\pi^{(1)}$ be the length-$|\eta_1|$ prefix of $\pi$, and define $\encode(\pi^{(1)}) \in \zo^{|\eta_1|}$ to be the string of values that $\pi^{(1)}$ assigns to the variables in $\eta_1$.\inote{We need at some point to say something about truncation, i.e., handle the annoying case when $s \neq |\eta_1| + \ldots + |\eta_{s'}|$.}

We now define $\sigma^{(i)}$ for a general $i\in [s']$. For ease of notation we define the restriction
$$\rho^{(i-1)} \eqdef \rho(\eta_1 \mapsto \pi^{(1)})\cdots (\eta_{i-1} \mapsto \pi^{(i-1)})
  \in \{0,1,\ast\}^{A(d)}.$$ 
Note that $\rho^{(0)}=\rho$. 
Recalling our $\mathsf{CanonicalDT}$ algorithm and the definition of $T_i$ as the $i$-th~term processed by $\mathsf{CanonicalDT}(F,\rho)$, we have that $T_i$ is the first term in $F$ such that $T_i \uhr \rho^{(i-1)}$~is~non-constant and $\smash{T_i \uhr\rho^{(i-1)}\rho' \equiv 1}$ for some $\smash{\rho' \in \supp(\calD^{(d)}_u)}$. Therefore, we have
\begin{align*}
 \eta_i =  \big\{\beta\in B(d) \colon x_{\beta,\tau} \text{~or~}\neg\,x_{\beta,\tau}\text{ occurs in $T_i \uhr \rho^{(i-1)}$ %\rho (\eta_1 \mapsto \pi^{(1)})\cdots (\eta_{i-1}\mapsto \pi^{(i-1)}))$ 
 for some {$\tau\in A'$}}\big\}.
 %\\
\end{align*}

\violet{We define $\sigma^{(i)} \in \{0,1,\ast\}^{B(d) \times A'}$ to be an arbitrary restriction (say the lexicographic first under the ordering $0 \prec 1 \prec \ast$) satisfying the following three properties: 
\begin{enumerate}
\item $T_i \uhr \rho^{(i-1)} \sigma^{(i)} \not\equiv 0$, and
\item $\sigma^{(i)} \in \supp(\calD^{(d)}_u)$, and 
\item $\sigma^{(i)}_{\beta} \in \{0,1\}^{A'}$ for all $\beta \in \eta_i$, and $\sigma^{(i)}_{\beta} = \{\ast\}^{A'}$ for all $\beta \notin \eta_i$. 
\end{enumerate} 
In words, $\sigma^{(i)}$ is the lexicographic first restriction in $\supp(\calD^{(d)}_u)$ that completely fixes blocks $\beta \in \eta_i$, leaves all other blocks $\beta \notin \eta_i$ free, and fixes the blocks in $\eta_i$ in a way that does not falsify $T_i \uhr \rho^{(i-1)}$. For $1 \le i < s'$, we recall that $\eta_i$ contains all blocks with variables occurring in $T_i \uhr \rho^{(i-1)}$, and so property (1) above can in fact be stated as $T_i \uhr \rho^{(i-1)} \sigma^{(i)} \equiv 1$.  (This is not necessarily true for the special case of $i = s'$ since $\eta_{s'}$ may only contain a subset of the blocks with variables occurring in $T_{s'} \uhr \rho^{(s'-1)}$; c.f.~(\ref{eq:special-case}).)}

We observe that such a restriction $\sigma^{(i)}$ (one satisfying all three properties above) must exist. As remarked at the start of this subsection, by the definition of $T_i$ there exists a restriction $\rho' \in \supp(\calD^{(d)}_u)$ such that~$T_i \uhr \rho^{(i-1)}\rho' \equiv 1$. \violet{This along with the fact that $\smash{\calD^{(d)}_u}$ is independent across blocks implies the existence of a restriction in $\supp(\calD^{(d)}_u)$ that fixes exactly the blocks in $\eta_i$ in a way that does not falsify $T_i \uhr \rho^{(i-1)}$. 
}

This finishes the definition of $\sigma^{(i)}$.
We record the following key properties of $\sigma^{(i)}$:

\begin{proposition}
\label{prop:sigma-satisfies} 
$T_i \uhr \rho^{(i-1)}\sigma^{(i)} \equiv 1$  \violet{for $1 \le i < s'$, and $T_{s'} \uhr \rho^{(s'-1)} \sigma^{(s')} \not\equiv 0$.} 
\end{proposition}

\begin{proposition}
\label{prop:weight-increase} 
For every $\beta \in \eta_i$, we have $\rho^{(i-1)}_{\beta}=\{*\}^{A'}$ %for all $\tau\in A'$ 
  whereas
%\[ \orange{ \in \{ \ast\}^{A(\ourd,\beta)}} \quad \text{whereas} \quad 
$\sigma^{(i)}_\beta \in \zo^{A'}$, and
%and furthermore, we have 
\[ \Prx_{\boldsymbol{\varrho}\leftarrow\calD_{u}(q)}[\hspace{0.03cm}
\boldsymbol{\varrho} = \rho^{(i-1)}_\beta\hspace{0.03cm}] = q \quad \text{whereas} \quad \Prx_{\boldsymbol{\varrho}\leftarrow\calD_{u}(q)}[\hspace{0.03cm} \boldsymbol{\varrho} = \sigma^{(i)}_\beta\hspace{0.03cm}] = \frac{1-q}{2u}.\] 
\ignore{\rnote{As before,  should it be
$\weight_{\calD_{u}}(\sigma^{(i)}_\beta) = \frac{1-q}{2u}$? \violet{{\bf Li-Yang:} Yup, fixed.}}}
\end{proposition} 
%Given that $\rho^{(i-1)}$ and $\rho^{(i-1)}

\subsection{Decodability}\label{sec:decodability}

\begin{lemma}
\label{lemma:injection}
The map\ignore{\rnote{Should the red thing in what follows be ``$[w] \times (\zo\times [w])^{\ourd-\ell}$'' instead?  Or can/should we replace the whole ``$(\zo\times [w])^{\ourd-\ell} \times A(\ell)$'' by simply ``$A^\dagger(\ourd)$'' instead -- all we want to get across is that the first coordinate of $\theta(\rho)$ is some restriction, right? \violet{{\bf Li-Yang:} Yup, fixed. What I had written previously was incorrect.}}}
 $\theta \colon \calB \to \{0,1,\ast\}^{A(\ourd)} \times \zo^s \times \zo^{s (\log r + 1)}$, where
\begin{equation}\label{maintheta}
 \theta(\rho) = \Big(\rho\sigma^{(1)}\cdots\sigma^{(s')},\hspace{0.07cm} \mathrm{binary}(\pi), \hspace{0.07cm}
  \encode(\eta_1)\circ \cdots\circ\encode(\eta_{s'})\Big), \end{equation}
is an injection.  
\end{lemma} 
%We will prove Lemma~\ref{lemma:injection} by showing that a decoder can recover $\rho$ given $(\rho\sigma, \orange{\mathrm{binary}}(\pi), \encode(\red{\eta}))$. We will assume inductively that the decoder has recovered the ``hybrid'' restriction 
%\[ \rho  (\eta_1 \mapsto \pi^{(1)}) \cdots (\eta_{i-1}\mapsto \pi^{(i-1)})\sigma^{(i)} \cdots \sigma^{(s')}~\orange{\text{and the corresponding sets}~\eta_1, \ldots, \eta_{i-1},}\] 
%the base case $i=1$ being $\rho\sigma^{(1)} \cdots \sigma^{(s')} = \rho\sigma$, which is trivially true by %assumption. 

We will prove Lemma~\ref{lemma:injection} by describing a 
  decoder that can recover $\rho\in \calB$ given 
  %$(\rho\sigma, \encode(\pi), \encode(\red{\eta}))$
  $\theta(\rho)$ as in (\ref{maintheta}). 
Let $\sigma=\sigma^{(1)}\cdots\sigma^{(s')}$.
Note that $s'$ can be derived from $\theta_3(\rho)$.
To obtain $\smash{\rho}$, it suffices to recover the sets $\smash{\eta_i}$,
  by simply replacing $\smash{(\rho\sigma)_{\beta,\tau}}$ with $\ast$ for 
  all $\smash{\beta\in \eta_1\cup\cdots\cup \eta_{s'}}$ and all $\tau\in A'$.
%We show that it suffices to decode $\eta_i$'s. %$\smash{\eta_1,\ldots,\eta_{s'}}$. %, by setting every variable in
  %block $\beta$ for some $\beta\in \eta_1\cup \cdots\cup \eta_{s'}$
  
To recover $\eta_i$'s, we assume inductively that the decoder has recovered the ``hybrid'' restriction 
\begin{equation}\label{eq:uiui}\rho^{(i-1)}\sigma^{(i)}\cdots\sigma^{(s')}= \rho  (\eta_1 \mapsto \pi^{(1)}) \cdots (\eta_{i-1}\mapsto \pi^{(i-1)})\hspace{0.03cm}\sigma^{(i)} \cdots \sigma^{(s')}~\text{and the  sets}~\eta_1, \ldots, \eta_{i-1},\end{equation}  
with the base case $i=1$ being $\rho\hspace{0.02cm}\sigma^{(1)} \cdots \sigma^{(s')}=\theta_1(\rho)$, which is trivially true by assumption. 
We~will show below how to decode 
  $T_i$ and $\eta_i$, and then obtain the next ``hybrid'' restriction 
$$
\rho^{(i )}\sigma^{(i+1)}\cdots\sigma^{(s')}= \rho  (\eta_1 \mapsto \pi^{(1)}) \cdots (\eta_{i }\mapsto \pi^{(i)})\hspace{0.03cm}\sigma^{(i+1)} \cdots \sigma^{(s')}~
$$ 
We can recover all $s'$ sets $\eta_1,\ldots,\eta_{s'}$ 
  after repeating this for $s'$ times. 

The following lemma shows how to recover $T_i$, given the ``hybrid'' restriction
  in (\ref{eq:uiui}).

%\inote{Following the previous inote, this seems to be true for all $i$ \emph{strictly less} than $s'$ only. I guess we can fix all this as follows. We mention that properties (1), (2), and (3) of $\sigma^{(i)}$ hold whenever $i < s'$. In the special case $i = s'$ we can try the following: we keep (2) and (3), and change (1) to ``$T_i \uhr \rho^{(i-1)} \sigma^{(i)} \equiv 1$ or $T_i \uhr \rho^{(i-1)} \sigma^{(i)}$ is non-zero and can be satisfied by yet another restriction in the support of our distribution'' (We should take the ``first'' restriction $\sigma^{(i)}$ with this property for decoding). In particular, we will not change $\eta_{s'}$, and set to 0/1 in $\sigma^{(i)}$ only the $\beta$'s in $\eta_{s'}$. Now the corresponding term $T_{s'}$ will not necessarily be set to true, but it seems to me that in the recovery procedure it is still the first term that is either 1 or non-constant and can be set to 1 by...}

\begin{proposition}
For $1 \le i < s'$, we have that $T_i$ is the first term in $F$ such that $$T_i \uhr \rho^{(i-1)}\sigma^{(i)} \cdots \sigma^{(s')} \equiv 1.$$
\violet{For the special case of $i = s'$, we have that $T_{s'}$ is the first term in $F$ such that 
\[ T_{s'} \uhr \rho^{(s'-1)} \sigma^{(s')} \rho'' \equiv 1 \]
 for some $\rho'' \in \supp(\calD^{(d)}_u)$.}  
  %Then $T_i' = T_i$, where $T_i$ is the $i$-th term processed by $\mathsf{CanonicalDT}(F,\rho)$ along the path $\pi$. }
\end{proposition} 
\begin{proof}
\violet{We first justify the claim for $1 \le i < s'$.}
Recall that $T_i$ is the first term in $F$ such that $T_i \uhr \rho^{(i-1)}$ is non-constant and $T_i \uhr \rho^{(i-1)}\rho' \equiv 1 $ for some restriction $\smash{\rho' \in \supp(\calD^{(d)}_u)}$.  This together with Proposition~\ref{prop:sigma-satisfies} implies that $T_i$ is the first term in $F$ such that $\smash{T_i \uhr \rho^{(i-1)}\sigma^{(i)} \equiv 1}$: as $\smash{\sigma^{(i)} \in \supp(\calD^{(u)}_d)}$, it follows that $\rho^{(i-1)}\sigma^{(i)}$ cannot satisfy any term that occurs before $T_i$ in $F$. For the same reason, $T_i$ remains the first term in $F$ such that $\smash{T_i \uhr \rho^{(i-1)}\sigma^{(i)}\cdots \sigma^{(s')} \equiv 1}$ (since $\smash{\sigma^{(i+1)},\cdots,\sigma^{(s')} \in \supp(\calD^{(d)}_u)}$
  and so is their composition).
  
  \violet{The argument for $i = s'$ is similar. We again recall that $T_{s'}$ is the first term in $F$ such that $T_{s'} \uhr \rho^{(s'-1)}$ is non-constant and $T_{s'} \uhr \rho^{(s'-1)}\rho' \equiv 1 $ for some restriction $\smash{\rho' \in \supp(\calD^{(d)}_u)}$. Since every term in $T$ that occurs before $T_{s'}$ in $F$ is such that $T \uhr \rho^{(s'-1)} \rho' \not\equiv 1$ for all $\rho' \in \supp(\calD^{(d)}_u)$, certainly  $T\uhr \rho^{(s'-1)} \sigma^{(s')} \rho'' \not\equiv 1$ for all $\rho'' \in \supp(\calD^{(d)}_u)$ as well. On the other hand, by Proposition~\ref{prop:sigma-satisfies} we have that $\sigma^{(s')}$ does not falsify $T_{s'} \uhr \rho^{(s'-1)}$, and so there must exist $\rho'' \in \supp(\calD^{(d)}_u)$ such that $T_{s'} \uhr \rho^{(s'-1)} \sigma^{(s')}\rho'' \equiv 1$. This completes the proof.} 
\end{proof}

With $T_i$ in hand we use $\encode(\eta_i)$ to reconstruct 
  $\eta_i$ by Proposition \ref{simpleprop}.
We modify the current ``hybrid'' restriction $\rho  (\eta_1 \mapsto \pi^{(1)}) \cdots (\eta_{i-1}\mapsto \pi^{(i-1)})\sigma^{(i)} \cdots \sigma^{(s')}$ as follows: for each $\beta \in \eta_i$, set 
\[ \big(\rho  (\eta_1 \mapsto \pi^{(1)}) \cdots (\eta_{i-1}\mapsto \pi^{(i-1)})\sigma^{(i)} \cdots \sigma^{(s')}\big)_{\beta,\tau} = \orange{\pi_\beta^{(i)}},\quad \text{for all $\tau \in A'$}.\] 
The resulting restriction is $\rho  (\eta_1 \mapsto \pi^{(1)}) \cdots (\eta_{i}\mapsto \pi^{(i)})\sigma^{(i+1)} \cdots \sigma^{(s')}$ as desired. 

Starting with $\rho\sigma$ and repeating this procedure for $s'$ times, we recover all $\eta_i$'s and then $\rho$.
This completes the proof that $\theta$ is an injection.

%Recurse until we get 
%\[ \rho  (\eta_1 \mapsto \pi^{(1)}) \cdots (\eta_{s'}\mapsto \pi^{(s')}),\]
%from which it is easy to recover $\rho$ using the transcript of the decoding process (or simply by %flipping all monochromatic blocks back to $\ast$.) 

\subsection{Weight increase}\label{sec:weight} 

Recall that $\rho$ and $\rho\sigma$ differ in exactly $s$ many blocks, and furthermore, $\rho$ is $\{ \ast\}^{A'}$ on all these blocks whereas $\rho\sigma$ belongs to $\zo^{A'} \cap \supp(\calD_{u})$ on these blocks. 

\begin{lemma}
\label{lemma:weight-increase} 
For any $\rho\in \calB$ and $\rho\sigma=\theta_1(\rho)$, we have
\[ \frac{\Pr[\hspace{0.03cm}\brho = \rho\sigma\hspace{0.03cm}]}{\Pr[\hspace{0.03cm}\brho = \rho\hspace{0.03cm}]}  = \mathop{\prod_{\text{blocks $\beta$ on}}}_{\text{which they differ}}  \frac{\Pr[\hspace{0.03cm}\boldsymbol{\varrho} = (\rho\sigma)_\beta\hspace{0.03cm}]}{\Pr[\hspace{0.03cm}\boldsymbol{\varrho} = \rho_\beta\hspace{0.03cm}]} = \bigg(\frac{1-q}{2q u}\bigg)^{\!s}. \]
\ignore{\rnote{Replace the 1 in numerator of RHS with $1-q$? \violet{{\bf Li-Yang:} Yup fixed.}}}
\end{lemma} 

\begin{proof}
This follows from independence across blocks and Proposition~\ref{prop:weight-increase}. 
\end{proof}

% !TEX root =  main.tex

\newcommand{\uth}{\bigskip \bigskip {\huge \violet{UP TO HERE}}\bigskip \bigskip}

\section{Proof of Theorem \ref{thm:our-depth-hierarchy}} \label{sec:proof_depth_hierarchy}

%In this section we use our projection switching lemma to prove our main  result, Theorem \ref{thm:main}.
In this section we prove our main technical result, Theorem \ref{thm:our-depth-hierarchy}, restated below:

\begin{reptheorem}{thm:our-depth-hierarchy}
\orange{There is an absolute constant $c>0$ such that the following holds. Let $d=d(w)$ and $u=u(w)$ satisfy $d\geq 1$ and $2 \leq u \leq w^{33/100}$. Then for $w$ sufficiently large, any depth-$d$ circuit computing the $\SkewedSipser_{u,d}$ function \emph{(}recall that this is a formula of depth $2d+1$ over $n=(uw)^dw^{33/100}$ variables\emph{)} must have size at least $n^{c \cdot (u/d)}$.}  
\end{reptheorem}
% OLD STATEMENT
%\begin{reptheorem}{thm:our-depth-hierarchy}
%For every $d = d(n)$ and $2 \le u\le w^{\red{33/100}}$, any depth-$d$ circuit computing the $\SkewedSipser_{u,d}$ function \emph{(}recall that this is a formula of depth $2d+1$ over $n=u^dw^{d+\red{33/100}}$ variables\emph{)} must have size $n^{\Omega(u/d)}$.  
%\end{reptheorem}

\ignore{\xnote{\red{I feel that it might be better to
  state the theorem for any $d\ge 1$ instead of saying ``for every $d=d(n)$.''
  The thing is that $n$ here is a function of $d$, which makes it 
  kind of different from Theorem 1 where $n$ is the number of vertices and it is natural
  to talk about $k$ and $d$ as functions of $n$.
For example, we can say ``for every $d\ge 1$ and $2\le u\le w^{33/100}$, ...''. Maybe 
  we also need to emphasize that $w$ should be sufficiently large.}}
  }

\red{We begin by first observing that the claimed $n^{\Omega(u/d)}$ circuit size lower bound is $o(n)$, and hence vacuous, if $d>u$;\ignore{\inote{Is an $\Omega(n)$ lower bound for all $d$ obvious? Recall that $\SkewedSipser_{u,d}$ is computed by depth $2d+1$ formulas of size $< n$ (we use number of gates instead of number of wires for circuit size -- otherwise the bound $\geq n$ would be trivial). \red{Rocco:  This is kind of what I want, for the bound $\geq n$ to be trivial :)  We never clearly state what the ``circuit size'' is, but I think we could think of it as the \# of gates where we count each input variable as a gate, no?  The only effect of adopting this convention is that it lets us (given the $\Omega(\cdot)$ in the exponent of our lower bounds) focus on the regime of super-polynomial size lower bounds, which is indeed the regime we're really interested in.  I would be OK with not having any discussion of this, but if you guys think we should have some discussion of this or use a different size convention then I'm okay with it.  I guess my goal is to avoid having us focus the reader's attention on these ``corner cases'' that are sort of peripheral to our main concerns.}  It should be enough to argue that size $n^{\Omega(1)}$ is required -- maybe we can easily argue this using gate elimination? (Adding the assumption that $d \leq u$ to the statement of this result seems to weaken our main lower bound for STCONN: I guess we may need to assume $d \leq \log k/\log \log k$ there.)}} thus it suffices to prove the claimed bound  under the assumption that $d \leq u.$  We make this assumption in the rest of the proof below (see specifically Corollary \ref{cor:preserve-target}).   Of course we can also assume that $d \geq 2$, since depth-1 circuits of any size cannot compute $\SkewedSipser_{u,1}.$
In the proof we set the parameter $q$ to be
  $\smash{q = w^{-669/1000}}$.}
  
%\red{Note that Theorem \ref{thm:our-depth-hierarchy} is trivial when 
%  $u<d$. So we always assume that $u\ge d$ in the rest of the section.}
%\ignore{
%\rnote{\red{The proof that Theorem \ref{thm:our-depth-hierarchy} yields Theorem
%\ref{thm:main} , and the  (not quite padding) argument giving the $n^{k^{\Omega(1/d)}}$ form of the lower %bound for $st$-connectivity for all $k$, will both go in the ``Reduction'' section.}}
%
%\medskip
%}
In Section \ref{sec:target-preservation} we establish that our target function $\SkewedSipser_{u,d}$ retains structure with high probability under a suitable random projection.  In Section \ref{sec:pf-lower-bound} we repeatedly apply both this result and  our projection switching lemma to prove Theorem \ref{thm:our-depth-hierarchy}.

\subsection{Target preservation} \label{sec:target-preservation}

We start with an easy proposition about what happens to \red{$\d2sipser_{u}$} %$\SkewedSipser_{u,d}$ 
  under a random restriction from
$\calD_{u}(q)$.  The following is an immediate consequence of Definition~\ref{def:D_ell} and \orange{Fact}~\ref{prop:column-kill}:

\begin{proposition}
\label{prop:preserve-target} 
For $\boldsymbol{\varrho}\leftarrow \calD_{u}(q)$, we have that 
\[ 
\big(\d2sipser_{u} \uhr\boldsymbol{\varrho}\big)\equiv
\begin{cases} 
\d2sipser_{u} & \text{with probability $q$} \\[0.3ex]
0 & \text{with probability $1-q$}. 
\end{cases}
\] 
\end{proposition} 

\red{We obtain the following corollary.}

\begin{corollary}
\label{cor:preserve-target} \ignore{\rnote{If we want to get the ``24'' down to $4.01$ or something we may want to set $\eps$ carefully and also perhaps change other aspects of the setup like having bottom fan-in $\sqrt{w}$.}}  For every $\red{1\le \ell \le d}$, we have that $\proj_\brho(\SkewedSipser_{u,\ell})$ 
  %\equiv \proj(\SkewedSipser_{u,\ell}\uhr \brho)$$ 
contains $\SkewedSipser_{u,\ell-1}$ as a subfunction with probability at least \violet{$0.9$} over 
  a random restriction $\smash{\brho\leftarrow \calD^{(\ell)}_{u}(q)}$. \end{corollary} 

\begin{proof}
Recall that  $\brho \leftarrow \calD^{(\ell)}_{u}(q)$ is drawn by independently drawing $\brho_\beta \leftarrow \calD_{u}(q)$  for each block~$\beta \in B(\ell)$. We have that 
  $\proj_\brho(\SkewedSipser_{u,\ell})$ contains $\SkewedSipser_{u,\ell-1}$ as a subfunction if the following holds:  for each of the $\smash{\OR^{(2)}}$ gates in $\SkewedSipser_{u,\ell}$, at least $\smash{w^{33/100}}$ of the $w$ $\smash{\AND^{(1)}}$~gates (each one corresponding to an independent $\d2sipser_{u}$ function) that are its children (say at addresses $\beta_1,\dots,\beta_w$) have $\brho_{\beta_i} \in \{*\}^{A'}$.  
  
  By Proposition \ref{prop:preserve-target}, 
  for a given $\OR^{(2)}$ gate, the expected number of $\beta_i$'s beneath it that have $\brho_{\beta_i} \in \{*\}^{A'}$ is $qw=w^{331/1000}$. So a multiplicative Chernoff bound shows that at least $w^{33/100}$ of the $\beta_i$'s beneath it have $\brho_{\beta_i} \in \{*\}^{A'}$ except with failure probability at most $e^{-w^{331/1000}/8}.$  \violet{By a union bound over the (at most $n$) $\OR^{(2)}$ gates in $\SkewedSipser_{u,\ell}$, we have that the overall failure probability is at most $n \cdot e^{-w^{331/1000}/8}$.  Since 
\[ n = u^d w^{d+ \frac{33}{100}} \le w^{\frac{133}{100}d +  \frac{33}{100}} \le w^{\frac{133}{100} u + \frac{33}{100}} \le w^{\frac{133}{100}w^{33/100} + \frac{33}{100}} \ll 0.1 \cdot e^{w^{331/1000}/8},\]
the proof is complete. (In the above we used $u \le w^{33/100}$ for the first inequality, $d \le u$ for the second, $u \le w^{33/100}$ again for the third, and \red{$w$ being sufficiently large for the last}.)  }
\end{proof}

\ignore{
\begin{corollary}
 \ignore{\rnote{If we want to get the ``24'' down to $4.01$ or something we may want to set $\eps$ carefully and also perhaps change other aspects of the setup like having bottom fan-in $\sqrt{w}$.}} Set $q = w^{-669/1000}$. For every $\ell \geq 1$, we have that $\proj_\brho(\SkewedSipser_{u,\ell}) \equiv \proj(\SkewedSipser_{u,\ell}\uhr \brho)$ contains $\SkewedSipser_{u,\ell-1}$ as a subfunction with probability at least \red{$1-n \cdot e^{-w^{\red{331/1000}}/8}$} over $\brho\leftarrow \calD^{(\ell)}_{u}(q)$. \end{corollary} 

\begin{proof}
Recall that  $\brho \leftarrow \calD^{(\ell)}_{u}(q)$ is drawn by independently drawing $\brho_\beta \leftarrow \calD_{u}$  for each $\beta \in B(\ell)$.
We have that $\proj(\SkewedSipser_{u,\ell}\uhr \brho)$ contains $\SkewedSipser_{u,\ell-1}$ as a subfunction \orange{if} the following holds:  for each of the (at most $n$) $\OR^{(2)}$ gates in $\SkewedSipser_{u,\ell}$, at least $\red{w^{33/100}}$ of the $w$ $\AND^{(1)}$ gates \orange{(each one corresponding to an independent $\d2sipser_{u}$ function)} that are its children (say at addresses $\beta_1,\dots,\beta_w$) have $\brho_{\beta_i} \in \{*\}^{A'}$.  \orange{By Proposition \ref{prop:preserve-target},} for a given $\OR^{(2)}$ gate, the expected number of $\beta_i$'s beneath it that have $\brho_{\beta_i} \in \{*\}^{A'}$ is $\red{qw=w^{331/1000}}$, so a multiplicative Chernoff bound gives that at least $\red{w^{33/100}}$ of the $\beta_i$'s beneath it have $\brho_{\beta_i} \in \{*\}^{A'}$ except with failure probability at most $e^{-w^{\red{331/1000}}/8}.$  The corollary follows by a union bound over the (at most $n$) $\OR^{(2)}$ gates in $\SkewedSipser_{u,\ell}.$
\end{proof}
}

\subsection{Completing the Proof of Theorem \ref{thm:our-depth-hierarchy}} \label{sec:pf-lower-bound}

%For such $d$ we have that $k^{1/d} \geq 2$ and hence $u=\lceil k^{1/d} \rceil = \Theta(k^{1/d}).$
Most of the proof is devoted to showing that the required size
  for a depth-$d$ circuit that computes $\SkewedSipser_{u,d}$ is at least 
\begin{equation} \label{eq:S}
S \eqdef {\frac {0.1}{\left(16u^{2}q\right)^{u-1}}}.
\end{equation}

\ignore{
\rnote{Should we say something about how we assume that the circuit is alternating, leveled, etc?  In \cite{RST:15} we had a paragraph (in a ``Notation'' subsection) which said in part

\red{``

A DNF is an $\OR$ of $\AND$s (terms) and a CNF is an $\AND$ of $\OR$s (clauses).  The \emph{width} of a DNF (respectively, CNF) is the maximum number of variables that occur in any one of its terms (respectively, clauses). We will assume throughout that our circuits are \emph{alternating}, meaning that every root-to-leaf path alternates between $\AND$ gates and $\OR$ gates, and \emph{layered}, meaning that for every gate $\mathsf{G}$, every root-to-{\sf G} path has the same length.  By a standard conversion, every depth-$d$ circuit is equivalent to a depth-$d$ alternating layered circuit with only a modest increase in size (which is negligible given the slack on our analysis).

''}}}

We prove (\ref{eq:S}) by contradiction; so assume there is a depth-$d$ circuit $C$ of size at most $S$ that computes $\SkewedSipser_{u,d}.$  As noted in Section~\ref{sec:preliminaries} we assume that $C$ is alternating and leveled.
%\red{We write $S_i$ to denote the number of gates in the level at distance  $i$ from the input variables (so in particular note that we have $S \geq S_1 + S_2 + \cdots + S_d)$.}\xnote{I found
%  it easier for me to understand the proof by using $S_i$'s instead of a single $S$ below, but it
%  should be easier to change it back.}

\violet{We ``get the argument off the ground'' by first hitting both $\SkewedSipser_{u,d}$ and $C$ with $\proj_\brho(\cdot)$ for $
\smash{\brho \leftarrow \calD^{(d)}_{u}(q)}$, where $q =w^{-669/1000}$. (By Remark \ref{rem:CNF_case}, we can apply our projection switching lemma, Theorem~\ref{theorem:PSL}, both to $r$-DNFs and $r$-CNFs.) Applying Theorem~\ref{theorem:PSL} (with $r=1$ and $s=u-1$) to each of the gates at distance 1 from the inputs in $C$,\footnote{In this initial application we view $C$ as having an extra layer of gates of fan-in $1$ next to the input variables, so we have a valid application of Theorem \ref{theorem:PSL} with $r = u - 1 \geq 1$.} we have that the resulting circuit $\proj_\brho(C)$ has depth $d$, bottom fan-in $u-1$,
  and \red{at most $S$ gates at distance at least $2$ from the inputs}\hspace{0.05cm}\footnote{\red{Note that $\proj_\brho(C)$
  may have a large number of gates at distance $1$ from the inputs but 
  it suffices for our purpose to bound the number of gates at distance at least
  $2$ from the inputs.}} with failure probability at most $S \cdot (16qu)^{u-1} < 0.1$.  On the other hand, taking $\ell=d$ in  Corollary~\ref{cor:preserve-target} we have that $\proj_{\brho}(\SkewedSipser_{u,\ourd})$ contains
$\SkewedSipser_{u,d-1}$ as a subfunction with failure probability at most $0.1$.   By a union bound, with probability at least $0.8$, a draw of $\smash{\brho \leftarrow \calD^{(d)}_u(q)}$ satisfies both of the above, and we fix any such restriction $\red{\restletter}^{(d)} \in \supp(\calD^{(d)}_u(q))$. A further deterministic ``trimming'' restriction (by only setting certain variables to $0$; note that this can only simplify $\proj_{\rho^{(d)}}(C)$ further) causes the target $\proj_{\rho^{(d)}}(\SkewedSipser_{u,d})$ to become exactly $\SkewedSipser_{u,d-1}$.  Let us write $C_d$ to denote the resulting simplified version of the original circuit $C$ after the combined ``project-and-trim''.
\red{As $C$ is supposed to compute $\SkewedSipser_{u,d}$,
  $C_d$ must compute $\SkewedSipser_{u,d-1}$.}
  
  %\lnote{There's a notational clash here with a different section here. We use $\rho^{(...)}$ to denote something different in the proof of the SL. \red{Xi: Maybe we replace it by $\restletter$?}}
  
Next, we consider what happens to $\SkewedSipser_{u,d-1}$ and $C_d$ if we hit them both with $\proj_\brho(\cdot)$ for $\smash{\brho \leftarrow \calD^{(\ourd-1)}_{u}(q)}$.  Applying Theorem~\ref{theorem:PSL} (with $r=s = u-1$) to each of the gates at distance $2$ 
  from the inputs and taking a union bound, % over the gates at distance 2 from the inputs, 
  the resulting circuit $\proj_\brho(C_d)$ has depth $(d-1)$, bottom fan-in $u-1$, and at most $S$ gates at distance at least $2$ from the inputs with failure~probability at most
$S \cdot (16 r u q)^{u - 1}  < S \cdot (16 q u^{2})^{u - 1} \le 0.1$. On the other hand, taking $\ell=d-1$ we can again apply Corollary~\ref{cor:preserve-target} to $\SkewedSipser_{u,d-1}$ and we have that $\proj_\brho(\SkewedSipser_{u,d-1})$ contains
$\SkewedSipser_{u,d-2}$ as a subfunction with failure probability at most $0.1$.  Once again by a union bound, with probability at least $0.8$ a draw of $\smash{\brho \leftarrow \calD^{(d-1)}_u(q)}$ satisfies both of the above, and we fix any such restriction $\red{\restletter}^{(d-1)} \in \supp(\calD^{d-1}_u(q))$. As before we perform a deterministic trimming restriction that causes the target $\proj_{\red{\restletter}^{(d-1)}}(\SkewedSipser_{u,d-1})$ to become exactly $\SkewedSipser_{u,d-2}$ and we let $C_{d-1}$ be the resulting simplified version of $C_d$ after the combined project-and-trim. As $C_d$ computes $\SkewedSipser_{u,d-1}$
  we have that $C_{d-1}$ must compute $\SkewedSipser_{u,d-2}$.

Repeating the argument above, each time taking $r=s=u-1$ in Theorem \ref{theorem:PSL},~there exist   a sequence of restrictions $\smash{\red{\restletter}^{(d-2)} \in \supp(\calD^{(d-2)}_u(q)),\ldots, \red{\restletter}^{(1)} \in \supp(\calD^{(1)}_u(q))}$ and their resulting circuits $C_{d-2}, \ldots, C_{1}$ such that 
\begin{itemize}
\item {\bf Hard function retains structure.} For $1 \le \ell \le d-2$, $\proj_{\red{\restletter}^{(\ell)}}(\SkewedSipser_{u,\ell})$ contains $\SkewedSipser_{u,\ell-1}$ as a subfunction, and hence there exists a deterministic trimming restriction that results in $\proj_{\red{\restletter}^{(\ell)}}(\SkewedSipser_{u,\ell})$ becoming exactly $\SkewedSipser_{u,\ell-1}$. 
\item {\bf Circuit collapses.} For $2 \le \ell \le d-2$, the circuit $\proj_{\red{\restletter}^{(\ell)}}(C_{\ell+1})$ has depth $\ell$, bottom fan-in $u-1$, and has at most $S$ gates at distance at least $2$ from the inputs. Furthermore, $C_\ell$ is the simplified version of $\proj_{\red{\restletter}^{(\ell)}}(C_{\ell+1})$ after the deterministic trimming restriction associated with $\proj_{\red{\restletter}^{(\ell)}}(\SkewedSipser_{u,\ell})$. \ignore{(defined in the bullet above). }Finally, the circuit $\proj_{\red{\restletter}^{(1)}}(C_2)$ can be expressed as a depth-$(u-1)$ decision tree, and $C_1$ is the  simplified version of $\proj_{\red{\restletter}^{(1)}}(C_2)$ after the deterministic trimming restriction associated with $\proj_{\red{\restletter}^{(1)}}(\SkewedSipser_{u,1})$. 
\end{itemize}

The above implies that $C_\ell$ computes $\SkewedSipser_{u,\ell-1}$ for all $1\le \ell \le d-2$. This yields the desired contradiction since $C_1$, a decision tree of depth at most $u-1$, cannot compute $\SkewedSipser_{u,0}$, the $\OR$ of  $\red{w^{33/100}\ge u}$ many variables. Hence any depth-$d$ circuit computing $\SkewedSipser_{u,d}$ must have size at least $S$, where $S$ is the quantity defined in (\ref{eq:S}). The following calculation showing that $\smash{S = n^{\Omega(u/d)}}$ completes the proof of Theorem~\ref{thm:our-depth-hierarchy}: 

\begin{claim} \label{claim:tech3}
$S = n^{\Omega(u/d)}$.
\end{claim}
\begin{proof}
We first observe that 
\[
n = u^d w^{d+\frac{33}{100}} \leq w^{\frac{133}{100}d + {\frac {33}{100}}} \leq w^{({\frac {133}{100}} + {\frac {33}{200}})d} < w^{{\frac {3d} 2}},
\quad \text{and hence} \quad
n^{{\frac 2 {3d}}} < w,
\]
where we used $u \leq w^{33/100}$ for the first inequality and $d \geq 2$ for the second.  As a result we have
\[
S = 0.1  \left({\frac {w^{669/1000}}{16u^2}}\right)^{u-1}
\geq 0.1  \left({\frac {w^{9/1000}}{16}}\right)^{u/2} \geq
0.1 \left({\frac {n^{{\frac 2 {3d}} \cdot {\frac 9 {1000}}}}{16}}\right)^{u/2}= n^{\Omega(u/d)},
\]
where we used $q=w^{-669/1000}$ for the first equality, $2 \leq u \leq w^{33/100}$ for the first inequality, and $w>n^{{\frac 2 {3d}}}$ for the final equality.
\end{proof}}

%\bigskip 

%\bigskip

%\bigskip 

\ignore{

We ``get the argument off the ground'' by first hitting both $\SkewedSipser_{u,d}$ and $C$ with $\proj_\brho(\cdot)$ for $
\brho \leftarrow \calD^{(d)}_{u}(q)$, where $q =w^{-669/1000}$ (as in Corollary \ref{cor:preserve-target}). Recall that, by Remark \ref{rem:CNF_case}, we can apply our Projection Switching Lemma (Theorem~\ref{theorem:PSL}) both to $r$-DNFs and $r$-CNFs. Applying Theorem~\ref{theorem:PSL} (with $r=1$ and $s=u-1$) to each of the \red{$S_1$ many} gates at distance 1 from the inputs in $C$, we have that with failure probability at most $\red{S_1} \cdot (16qu)^{u - 1}$ the resulting circuit $\proj(C \uhr \brho)$ now has bottom fan-in at most $u-1$ (note that it may still have depth $d$).\footnote{{Put another way, in this initial application we view $C$ as having an extra layer of gates of fan-in $1$ next to the input variables, so we have a valid application of Theorem \ref{theorem:PSL} with $r = u - 1 \geq 1$.}} On the other hand, taking $\ell=d$ in  Corollary~\ref{cor:preserve-target} we have that $\proj(\SkewedSipser_{u,\ourd} \uhr \brho)$ still contains
$\SkewedSipser_{u,d-1}$ as a subfunction with failure probability at most $\red{n \cdot e^{-w^{331/1000}/8}}$.  Assuming that the failure event does not take place, a further deterministic ``trimming'' restriction (only setting variables to 0; note that this can only simplify $C$ further) causes the target $\SkewedSipser_{u,d}$ to become exactly $\SkewedSipser_{u,d-1}$.  Let us write $C_0$ to denote the resulting simplified version of the original circuit $C$.

Next we consider what happens to both $\SkewedSipser_{u,d-1}$ and $C_0$ if we hit them both with $\proj_\brho(\cdot)$ for $\brho \leftarrow \calD^{(\ourd-1)}_{u}(q)$.  Applying Theorem~\ref{theorem:PSL} (with $r=s = u-1$) and taking a union bound over the \red{$S_2$ many} gates at depth 2 from the inputs, with failure probability at most
$\red{S_2} \cdot (16 r u q)^{u - 1}  < \red{S_2} \cdot (16 q u^{2})^{u - 1}$, the resulting circuit $\proj(C_0 \uhr \brho)$ has depth \orange{at most} $d-1$ and bottom 
fan-in at most  $ u-1$.\lnote{We need to bound its size as well, right? Its size may blow up, but I think the key thing is that the number of gates at distance at least $2$ from the inputs will remain $\le S$. (In fact, it is at most $S_2 + \cdots + S_d$, but I don't think this saves us anything.)}  On the other hand, taking $\ell=d-1$ we can again apply Corollary~\ref{cor:preserve-target} to $\SkewedSipser_{u,d-1}$ and we have that $\proj(\SkewedSipser_{u,d-1} \uhr \brho)$ contains
$\SkewedSipser_{u,d-2}$ as a subfunction with failure probability at most $\red{n \cdot e^{-w^{331/1000}/8}}$.  As before we can perform a deterministic ``trimming'' step to cause the target function to become exactly
$\SkewedSipser_{u,d-2}$.  Let $C_1$ be the resulting simplified version of $C_0$.

We repeat this ``project-and-trim'' step a total of $d-1$ times to obtain circuits $C_1,\dots,C_{d-1}$, where at each application we take $r=s=u-1$ in Theorem \ref{theorem:PSL} as above.  Let $\boldsymbol{\Psi}$ denote the composition of all $d$ ``project-and-trim'' steps (including the very first one that was used to reduce the bottom fan-in of $C$). The upshot is that the final
circuit $C_{d-1} = \proj(C \uhr \boldsymbol{\Psi})$\lnote{To use this notation, it seems that we need to explain how one can view $\mathbf{\Psi}$, a sequence of project-and-trim operations over different variable spaces, as a restriction over the original variable space. I tried to sidestep this in my attempt in violet above.} can be expressed as a decision tree of depth $u-1$ except with overall failure probability at most\ignore{\inote{It doesn't matter much but I believe there is no need to have a factor of $d$ in the union bound: we can fix a good random projection during each trimming step instead of considering the event that all trimming steps must be good. We can leave this unchanged.  \red{Rocco:}  I don't think the ``$d$'' has anything to do with the trimming, the trimming is deterministic.  You are saying, the one factor of ``$S$'' accounts for all of the applications
of the union bound and there is no need for the extra ``$d$'', right?  I agree with this; if we want to get it across, we should introduce new parameters (perhaps $S_2, S_3$, $\dots$ where $S_i$ is the number of gates in the original circuit at distance $i$ from the inputs) and use them
in the above exposition instead of $S$.}}\inote{Sorry for not explaining what I meant last time in more detail. I guess here we are considering different probability spaces, one during each depth reduction. When we reduce the depth of the approximator and simplify the target, we should fix a ``good'' random projection/restriction, since formally the next application requires the correct set of variables in each DNF/CNF (based on the statement of our PSL -- maybe we should use $\ell$ instead of $d$ there in $A(d)$?). So we will be in good shape during the $i$-th stage/union bound as long as $S_i \cdot (16qu)^{u-1} + n \cdot e^{-w^{\red{331/1000}}/8} < 1$, which means we can fix an appropriate restriction, project, and continue with the argument. Using our analysis it is sufficient that $\max \{S_1, \ldots, S_d\} \cdot  (16qru)^{u-1} \leq S \cdot  (16qru)^{u-1} \leq 0.1$, and there is almost nothing that needs to be updated in the text. (We can also remove the factor of $d$ in the next equation.) Let me know if this makes sense.}
\begin{align*} \red{S_1} \cdot (16qu)^{u-1} + (\red{S_2 + \cdots + S_d}) \cdot (16 qru)^{u - 1} 
\le
 S \cdot (16 qu^{2})^{u - 1} \leq 0.1,
 \end{align*}
where the last inequality is by our choice of $S$.  On the other hand, for each of the values $\ell=d-1,\dots,1$ we apply Corollary \ref{cor:preserve-target} to $\SkewedSipser_{u,\ell}$, and by a union bound we get that $\proj(\SkewedSipser_{u,d} \uhr \boldsymbol{\Psi})$ is exactly $\SkewedSipser_{u,0}$ except with total failure probability at most 
\[
d \cdot n \cdot e^{-w^{\red{331/1000}}/8}.
\]
In Claim \ref{claim:tech1} below we show that $d \cdot n \cdot e^{-w^{\red{331/1000}}/8} < 0.9.$
Thus there is an outcome $\Psi$ of $\boldsymbol{\Psi}$ such that both of the following hold:  (i) $C_{d-1}$ computes a depth-$(u-1)$ decision tree, and (ii) $\proj(\SkewedSipser_{u,d} \uhr {\Psi})$ computes the $\SkewedSipser_{u,0}$ function, i.e. an $\OR$ of $\red{w^{33/100}}$ variables.  Since $u \leq \red{w^{33/100}}$ (by the hypothesis of the theorem), this means that it is impossible for a depth-$(u-1)$ decision tree to compute an $\OR$ of $\red{w^{33/100}}$ variables, and we have a contradiction to our initial assumption that the size-$S$, depth-$d$ circuit $C$ computes $\SkewedSipser_{u,d}.$  Thus we have established (\ref{eq:S}):  any depth-$d$ circuit for the $n$-variable $\SkewedSipser_{u,d}$ function must have size at least
\[
S:= {\frac {0.1}{ \left(16u^{2}q\right)^{u-1}}}.
\]
Finally, Claim \ref{claim:tech3} below shows that $S = n^{\Omega(u/d)}$, and this concludes the proof of Theorem \ref{thm:our-depth-hierarchy} (modulo Claims \ref{claim:tech1} and \ref{claim:tech3}).\qed

\gray{
\begin{claim} \label{claim:tech1}\lnote{Combined this with Corollary~\ref{cor:preserve-target}}
$d \cdot n \cdot e^{-w^{\red{331/1000}}/8} < 0.9.$
\end{claim}
\begin{proof}
We have
\begin{align*}
d \cdot n < n^2 = u^{2d}w^{2d+\red{{\frac {33}{50}}}} \leq w^{\red{{\frac {133}{50}}d + {\frac {33}{50}}}} < w^{\red{{\frac {133}{50}}u + {\frac {33}{50}}}} \leq w^{\red{{\frac {133}{50}}w^{33/100} + {\frac {33}{50}}}}
\ll 0.9 \cdot e^{w^{\red{331/1000}}/8}
\end{align*}
for $w$ sufficiently large.  (We used $d < n$ for the first inequality, $u \leq w^{\red{33/100}}$ for the second, $d \leq u$
for the third and $u \leq w^{\red{33/100}}$ again for the fourth.)
\end{proof}
}

\begin{claim} \label{claim:tech3}
$S = n^{\Omega(u/d)}$.
\end{claim}
\begin{proof}
We first observe that 
\[
n = u^d w^{d+\red{{\frac {33}{100}}}} \leq w^{\red{{\frac {133}{100}}d + {\frac {33}{100}}}} \leq \red{w^{({\frac {133}{100}} + {\frac {33}{200}})d}} < w^{{\frac {3d} 2}},
\quad \text{and hence} \quad
n^{{\frac 2 {3d}}} < w,
\]
where we used $u \leq w^{\red{33/100}}$ for the first inequality and $d \geq 2$ for the second.  As a result we have
\[
S = 0.1 \cdot \left({\frac {w^{\red{669/1000}}}{16u^2}}\right)^{u-1}
\geq 0.1 \cdot \left({\frac {w^{\red{9/1000}}}{16}}\right)^{u/2} \geq
{\frac 1 {10}} \cdot \left({\frac {n^{{\frac 2 {3d}} \cdot {\frac 9 {1000}}}}{16}}\right)^{u/2}= n^{\Omega(u/d)},
\]
where we used $\red{q=w^{-669/1000}}$ for the first equality, $2 \leq u \leq w^{\red{33/100}}$ for the first inequality, and $w>n^{{\frac 2 {3d}}}$ for the final equality.
\end{proof}}

\begin{remark} \label{remark:upperbound}
We remark that a straightforward construction yields small-depth circuits computing $\SkewedSipser_{u,d}$ that nearly match the lower bound given by Theorem \ref{thm:our-depth-hierarchy}.  This construction simply applies de Morgan's law to convert a $u$-way $\AND$ of $w$-way $\OR$s into a $w^u$-way $\OR$ of $u$-way $\AND$s.  This is done for all of the $\AND^{(d)},\AND^{(d-2)},\AND^{(d-4)},\dots$ gates in $\SkewedSipser_{u,d}$.  Collapsing adjacent layers of gates after this conversion, we obtain a depth-$(d+1)$ circuit of size 
$\smash{\violet{n^{O(u/d)}}}$ that computes the $\SkewedSipser_{u,d}$ function.
\end{remark}

\bibliographystyle{alpha}	
\bibliography{refs}

\begin{thebibliography}{KPPY84}

\bibitem[Ajt83]{ajtai1983}
Mikl{\'o}s Ajtai.
\newblock {$\Sigma_1^1$}-formulae on finite structures.
\newblock {\em Annals of {P}ure and {A}pplied {L}ogic}, 24(1):1--48, 1983.

\bibitem[Ajt89]{Ajtai:89}
Mikl\'os Ajtai.
\newblock First-order definability on finite structures.
\newblock {\em Annals of Pure and Applied Logic}, 45:211--225, 1989.

\bibitem[Bea95]{beame1993switching}
Paul Beame.
\newblock A switching lemma primer. {University of Washington, Dept. of
  Computer Science and Engineering, Technical Report UW-CSE-95-07-01}, 1995.

\bibitem[BIP98]{BIP:98}
Paul Beame, Russell Impagliazzo, and Toniann Pitassi.
\newblock Improved depth lower bounds for small distance connectivity.
\newblock {\em Computational Complexity}, 7:325 --345, 1998.

\bibitem[BPU92]{BPU:92}
Stephen Bellantoni, Toniann Pitassi, and Alasdair Urquhart.
\newblock Approximation and small depth {F}rege proofs.
\newblock {\em SIAM Journal on Computing}, 21(6):1161--1179, 1992.

\bibitem[Cai86]{Cai86}
Jin{-}Yi Cai.
\newblock With probability one, a random oracle separates {{\sf {PSPACE}}} from
  the polynomial-time hierarchy.
\newblock In {\em Proceedings of the 18th Annual {ACM} Symposium on Theory of
  Computing (STOC)}, pages 21--29, 1986.

\bibitem[FSS84]{FSS:84}
Merrick Furst, James Saxe, and Michael Sipser.
\newblock Parity, circuits, and the polynomial-time hierarchy.
\newblock {\em Mathematical Systems Theory}, 17(1):13--27, 1984.

\bibitem[H{\aa}s86]{Hastad:86}
Johan H{\aa}stad.
\newblock {\em Computational Limitations for Small Depth Circuits}.
\newblock MIT Press, Cambridge, MA, 1986.

\bibitem[IPS97]{IPS:97}
Russell Impagliazzo, Ramamohan Paturi, and Michael~E. Saks.
\newblock Size--depth tradeoffs for threshold circuits.
\newblock {\em SIAM Journal on Computing}, 26(3):693--707, 1997.

\bibitem[IS01]{IS01}
Russell Impagliazzo and Nathan Segerlind.
\newblock Counting axioms do not polynomially simulate counting gates.
\newblock In {\em Proceedings of the 42nd Annual IEEE Symposium on Foundations
  of Computer Science (FOCS)}, pages 200--209, 2001.

\bibitem[KPPY84]{KPNY84}
Maria Klawe, Wolfgang Paul, Nicholas Pippenger, and Mihalis Yannakakis.
\newblock On monotone formulae with restricted depth.
\newblock In {\em {Proceedings of the 16th Annual ACM Symposium on Theory of
  Computing (STOC)}}, pages 480--487, 1984.

\bibitem[Ros14]{Rossman:13}
Benjamin Rossman.
\newblock Formulas vs. circuits for small distance connectivity.
\newblock In {\em Proceedings of the 46th Annual ACM Symposium on Theory of
  Computing (STOC)}, pages 203--212. ACM, 2014.

\bibitem[RST15]{RST:15}
Benjamin Rossman, Rocco~A. Servedio, and Li-Yang Tan.
\newblock An average-case depth hierarchy theorem for {B}oolean circuits.
\newblock In {\em Proceedings of the 56th Annual IEEE Symposium on Foundations
  of Computer Science (FOCS)}, 2015.
\newblock to appear.

\bibitem[Sav70]{Savitch:70}
Walter Savitch.
\newblock Relationships between nondeterministic and deterministic tape
  complexities.
\newblock {\em Journal of Computer and System Sciences}, 4:177--192, 1970.

\bibitem[Sip83]{sipser1983}
Michael Sipser.
\newblock Borel sets and circuit complexity.
\newblock In {\em Proceedings of the 15th Annual ACM Symposium on Theory of
  Computing (STOC)}, pages 61--69, 1983.

\bibitem[Sta]{stackexchange-layering-circuits}
Theoretical Computer~Science StackExchange.
\newblock Available at\\
  \url{http://cstheory.stackexchange.com/questions/7672/most-efficient-way-to-convert-an-textac0-circuit-to-a-circuit-of-any-dep}.

\bibitem[Wig92]{Wigderson:92connectivity}
Avi Wigderson.
\newblock The complexity of graph connectivity.
\newblock In {\em Proceedings of the 17th Symposium on Mathematical Foundations
  of Computer Science (MFCS)}, pages 112--132. Springer-Verlag, 1992.

\bibitem[Yao85]{Yao:85}
Andrew Yao.
\newblock Separating the polynomial time hierarchy by oracles.
\newblock In {\em Proceedings of the 26th Annual IEEE Symposium on Foundations
  of Computer Science (FOCS)}, pages 1--10, 1985.

\end{thebibliography}

\end{document}